\documentclass[11pt]{article}
\usepackage[dvipsnames,svgnames,x11names,hyperref]{xcolor}
\usepackage{amsmath}
\usepackage{amsthm}
\usepackage[noend]{algpseudocode}
\usepackage{algorithm}
\usepackage{algorithmicx}
\usepackage[T1]{fontenc}
\usepackage{comment}
\usepackage{fullpage}
\usepackage{caption}
\usepackage{thmtools}
\usepackage{stmaryrd}
\usepackage{makecell}
\usepackage[colorlinks]{hyperref}
\hypersetup{
  linkcolor=[rgb]{0,0,0.4},
  citecolor=[rgb]{0, 0.4, 0},
  urlcolor=[rgb]{0.6, 0, 0}
}
\usepackage{amssymb,amsfonts,amsmath,amsthm,amscd,dsfont,mathrsfs,bbm}
\usepackage{todonotes}
\usepackage{complexity}
\usepackage{tikz}
\usetikzlibrary{decorations.pathreplacing,backgrounds}
\usepackage{multirow}
\usepackage[letterpaper, margin=1in]{geometry}
\usepackage{float}
\usepackage{amsthm}
\usepackage{thmtools,thm-restate}

\numberwithin{equation}{section}
\declaretheoremstyle[bodyfont=\it,qed=\qedsymbol]{noproofstyle}

\declaretheorem[name=Observation,numbered=no]{observation*}

\declaretheorem[numberlike=equation]{theorem}

\declaretheorem[name=Theorem,numbered=no]{theorem*}

\declaretheorem[numberlike=equation]{lemma}
\declaretheorem[name=Lemma,numbered=no]{lemma*}

\declaretheorem[name=Corollary,numbered=no]{corollary*}

\declaretheorem[numberlike=equation]{proposition}
\declaretheorem[name=Proposition,numbered=no]{proposition*}

\declaretheorem[numberlike=equation]{claim}
\declaretheorem[name=Claim,numbered=no]{claim*}

\declaretheorem[name=Conjecture,numbered=no]{conjecture*}

\declaretheorem[name=Question,numbered=no]{question*}

\declaretheoremstyle[bodyfont=\it,qed=$\lozenge$]{defstyle} 

\declaretheorem[numberlike=equation,style=defstyle]{definition}
\declaretheorem[unnumbered,name=Definition,style=defstyle]{definition*}

\declaretheorem[unnumbered,name=Example,style=defstyle]{example*}

\declaretheorem[unnumbered,name=Notation=defstyle]{notation*}

\declaretheorem[unnumbered,name=Construction,style=defstyle]{construction*}

\declaretheorem[unnumbered,name=Remark,style=defstyle]{remark*}

\newcommand{\F}{\mathbb{F}}
\newcommand{\N}{\mathbb{N}}

\newcommand{\Z}{\mathbb{Z}}

\newcommand{\va}{\mathbf{a}}

\newcommand{\ve}{\mathbf{e}}
\newcommand{\vx}{\mathbf{x}}

\newcommand{\coeff}{\operatorname{Coeff}}
\newcommand{\ideal}[1]{(#1)}

\newcommand{\vb}{\mathbf{b}}

\newcommand{\var}[1]{\mathbf #1}

\newcommand{\vece}{\mathbf{e}}
\newcommand{\vecz}{\mathbf{z}}
\newcommand{\vecx}{\mathbf{x}}

\newcommand{\modulo}[1]{({\mathbb Z}/#1{\mathbb Z})}
\newcommand{\nbmodulo}[1]{\mathbb Z/#1{\mathbb Z}}
\newcommand{\der}[2]{#1^{(#2)}}

\newcommand{\hpartial}{\overline{\partial}}

\DeclareMathOperator*{\Ex}{\mathbb{E}}

\newcommand{\Mnote}[1]{\textcolor{Red}{\guillemotleft MK: #1\guillemotright}}

\newcommand{\algwidth}{\linewidth}

\newif\ifshowauthors
\showauthorstrue

\title{Fast Multivariate Multipoint Evaluation Over All Finite Fields}

\ifshowauthors
\author{
Vishwas Bhargava\thanks{Department of Computer Science, Rutgers University, Piscataway, NJ 08854. Research supported in part by the Simons 
Collaboration on Algorithms and Geometry and NSF grant CCF-1909683. Email: \texttt{vishwas1384@gmail.com}}
\and 
Sumanta Ghosh\thanks{Department of Computing and Mathematical Sciences, Caltech. Email: \texttt{besusumanta@gmail.com}}
\and 
Zeyu Guo\thanks{Department of Computer Science, UT Austin. Email: \texttt{zguotcs@gmail.com}}
\and 
Mrinal Kumar\thanks{Department of Computer Science \& Engineering, IIT Bombay. Email: \texttt{mrinal@cse.iitb.ac.in}}
\and
Chris Umans\thanks{Department of Computing and Mathematical Sciences, Caltech. Email: \texttt{umans@cs.caltech.edu}}
}
\else
\author{}
\fi

\date{}

\begin{document}

\maketitle
\pagenumbering{gobble}
\abstract{
Multivariate multipoint evaluation is the problem of evaluating a multivariate polynomial, given as a coefficient vector, simultaneously at multiple evaluation points.  
In this work, we show that there exists a deterministic algorithm for multivariate multipoint evaluation over any finite field $\F$ that outputs the evaluations of an $m$-variate polynomial of degree less than $d$ in each variable at $N$ points in time 
\[
(d^m+N)^{1+o(1)}\cdot\poly(m,d,\log|\F|)
\]
for all $m\in\N$ and all sufficiently large $d\in\N$.

A previous work of Kedlaya and Umans (FOCS 2008, SICOMP 2011)
achieved the same time complexity when the number of variables $m$ is at most $d^{o(1)}$ and had left the problem of removing this condition as an open problem. A recent work of Bhargava, Ghosh, Kumar and Mohapatra (STOC 2022)  
answered this question when the underlying field is not \emph{too} large  and has characteristic less than $d^{o(1)}$. In this work, we remove this constraint on the number of variables over all finite fields, thereby answering the question of Kedlaya and Umans over all finite fields.

Our algorithm relies on a non-trivial combination of ideas from three seemingly different previously known algorithms for multivariate multipoint evaluation, namely the algorithms of Kedlaya and Umans, that of  Bj\"orklund, Kaski and Williams (IPEC 2017, Algorithmica 2019),
and that of Bhargava, Ghosh, Kumar and Mohapatra, together with a result of Bombieri and Vinogradov from analytic number theory about the distribution of primes in an arithmetic progression.

We also present a second algorithm for multivariate multipoint evaluation that is completely elementary and in particular, avoids the use of the Bombieri--Vinogradov Theorem. However, it requires a mild assumption that the field size is bounded by an exponential-tower in $d$ of bounded \textit{height}. More specifically, our second algorithm solves the multivariate multipoint evaluation problem over a finite field $\F$ in time 
\[
(d^m+N)^{1+o(1)}\cdot\poly(m,d,\log |\F|)\]
for all $m\in\N$ and all sufficiently large $d\in\N$, 
provided that the size of the finite field $\F$ is at most $(\exp(\exp(\exp(\cdots (\exp(d)))))$, where the height of this tower of exponentials is fixed.
}
\newpage

\tableofcontents

\newpage
\pagenumbering{arabic}
\section{Introduction}
We study the problem of multivariate multipoint evaluation: given an $m$-variate polynomial $f(\vx) \in \F[\vx]$ of  degree less than $d$ in each variable, and $N$ points $\va_1, \va_2, \ldots, \va_N \in \F^m$, output $f(\va_1), f(\va_2), \ldots, f(\va_N)$. Here $\F$ is the underlying field. The input polynomial $f(\vx)$ is given by its coefficient vector. Therefore, the overall input can be represented by a list of $(d^m+mN)$ elements in $\F$. A 
trivial algorithm for this problem is to evaluate $f(\vx)$ at each $\va_i$ separately. Since evaluating $f(\vx)$ at each $\va_i$ takes $d^m\cdot\poly(d,m)$ operations over $\F$, this algorithm needs $Nd^m\cdot\poly(d,m)$ $\F$-operations in total. For $N=\Theta(d^m)$, the time complexity of this algorithm is quadratic with respect to the input size. Therefore, a natural algorithmic question here is to seek faster algorithms for this problem. Of particular interest would be to have an algorithm for this problem whose time complexity is \emph{nearly linear}, more specifically $(d^m+N)^{1+o(1)}$ (multiplied by lower-order $\poly(d, n, \log |\F|)$ terms), with respect to the input size.  

In addition to its innate appeal as a fundamental and natural question in computational algebra, fast algorithms for multivariate multipoint evaluation are closely related to fast algorithms for other important algebraic problems such as polynomial factorization and modular composition. For a detailed discussion on these connections, we refer to the work of Kedlaya and Umans \cite{Kedlaya11}. In a recent work, Bhargava, Ghosh, Kumar and Mohapatra \cite{BGKM21} used the special structure of their algorithm for multivariate multipoint evaluation to show an upper bound on the rigidity of Vandermonde matrices and very efficient algebraic data structures for the polynomial evaluation problem over finite fields.



For the setting of univariate polynomials, Borodin and Moenck \cite{BM74} showed that  the multipoint evaluation can be solved in nearly linear time. Their algorithm is short, simple and elementary, and proceeds via an application of the Fast Fourier Transform (FFT). However, this approach does not seem to extend when the number of variables exceeds one; in fact, even when the number of variables is two. However, for the multivariate case, when the input points form a product set, one can naturally extend the ideas in Borodin and Moenck \cite{BM74} to  get a nearly linear time algorithm for this problem. But, when the input points are arbitrary, getting a sub-quadratic algorithm for multipoint evaluation seems to be significantly more difficult. In fact, about three decades after Borodin and Moenck's work, N\"usken and Ziegler \cite{NZ04} proved that multipoint evaluation can be solved in most $O(d^{{\omega_2}/2 + 1})$ operations for $m= 2$ and $N = d^2$, where $\omega_2$ is the exponent for multiplying a $d\times d$ and a $d\times d^2$ matrix. The work \cite{NZ04} extends to general $m$ and gives an algorithm for multipoint evaluation that performs  $O(d^{{\omega_2}/2\cdot (m-1)+1})$ field operations. 

Two significant milestones in  this line of work are the results of Umans \cite{Umans08} and Kedlaya and Umans \cite{Kedlaya11}. Umans \cite{Umans08} gave a nearly linear time (that is, $(d^m+N)^{1+o(1)}\cdot \poly(m,d,\log|F|)$-time) \emph{algebraic} algorithm\footnote{Algorithms for multivariate multipoint evaluation can be divided into two categories: (1) \emph{algebraic algorithms}, where we are only allowed to perform arithmetic operations over the underlying field $\F$, and (2) \emph{non-algebraic algorithms}, where we are allowed to do bit operations. The algorithms of Kedlaya and Umans \cite{Kedlaya11} and those in this paper are not algebraic, whereas the algorithm of Borodin and Moenck \cite{BM74}, that of Umans \cite{Umans08}, and that of Bhargava, Ghosh, Kumar and Mohapatra \cite{BGKM21} are all algebraic. } for this problem over finite fields, provided that the characteristic of the field and the number of variables are at most  $d^{o(1)}$. Later, Kedlaya and Umans \cite{Kedlaya11} gave a nearly linear time \emph{non-algebraic} algorithm for \emph{all} finite fields, but they also need $m=d^{o(1)}$. In a recent work, Bhargava, Ghosh, Kumar and Mohapatra  \cite{BGKM21} improve the result of Umans \cite{Umans08} by removing the restriction on $m$ over finite fields whose characteristics are small and sizes are not too large. More specifically, they gave a nearly linear time algebraic algorithm for multivariate multipoint evaluation, provided that the characteristic of the field is $d^{o(1)}$  and the size of the field is at most $(\exp(\exp(\exp(\cdots (\exp(d)))))$, where the height of this tower of exponentials is fixed.  Another closely related result is a recent work of Bj\"orklund, Kaski and Williams \cite{BKW19} who (among other results) gave an algorithm for multivariate multipoint evaluation, but their time complexity depends polynomially on the field size (and not polynomially on the logarithm of the field size), and instead of $d^m$, their time complexity is nearly linear in $D^m$ where $D$ is the total degree of the polynomial. Nevertheless, their results play a crucial role in proving the results of this paper and we will discuss them in more detail in \autoref{sec: BKW overview}.  

Thus, from the context of previous work, a very natural and interesting open question is to design an algorithm, for the problem of multivariate multipoint evaluation that runs in nearly linear time and works for all finite fields and all ranges of the number of variables. Indeed, Kedlaya and Umans \cite{Kedlaya11} mention this as an open problem. 

In this work, we answer this question by giving two different algorithms for multivariate multipoint evaluation over finite fields. While our first algorithm works over all finite fields, the second algorithm still requires that the field size is not too large in terms of $d$. We now state our results and discuss the pros and the cons of the two algorithms and compare them to the algorithms known in prior work. Both our algorithms happen to be non-algebraic, i.e. we need more than just arithmetic operations over the underlying field. 

We now state our results and discuss how they compare with prior work. 

\subsection{Our Results}

We state our main result as follows. 
\begin{theorem}\label{thm: main informal}
There is a deterministic algorithm that given the coefficient vector of an $m$-variate polynomial $f(\vx)$ of degree less than $d$ in each variable over a finite field $\F$ and $N$ points
$\va_1,\va_2,\dots,\va_N\in \F^m$, outputs $f(\va_1),f(\va_2),\dots,f(\va_N)$ in time
\[
 (d^m + N)^{1 + o(1)}\cdot \poly(m, d, \log |\F|),  
\]
for all $m \in \N$ and all sufficiently large $d \in \N$.
\end{theorem}

\begin{remark*}
Throughout this paper, when we say $d$ is \emph{sufficiently large}, it means $d=\omega(1)$. 
\end{remark*}

The proof of the above theorem crucially relies on a deep result from analytic number theory, known as the {Bombieri--Vinogradov Theorem} \cite{bombieri65, vinogradov1965density}, related to the distribution of primes in arithmetic progressions. 

We also give a different algorithm for multivariate multipoint evaluation that avoids the Bombieri--Vinogradov Theorem and is completely elementary, but it requires the size of the finite field to be not too large: at most $(\exp(\exp(\exp(\cdots (\exp(d)))))$, where the height of this tower of exponentials is fixed. In other words, it removes the restriction on the characteristic of the field in the work of Bhargava \emph{et.al.} \cite{BGKM21}, but via a non-algebraic algorithm. 

We remark that neither of our algorithms is algebraic, and in particular, we crucially rely on working with the bit representation of the inputs. To obtain an algebraic algorithm for multivariate multipoint evaluation, for large $m$, and over all finite fields is a fundamental algebraic problem that continues to remain open. 

Below, in \autoref{tab:comparison}, we compare the results in this paper with the previously known results. 

\begin{table}[H]
\centering
\caption{Comparison with Prior Results}
\begin{tabular}{|c|c|c|c|c|}
 \hline
 \multicolumn{5}{|c|}{Multivariate Multipoint Evaluation over a Finite Field $\F_{q}$ of Characteristic $p$} \\
 \hline
 Results & Time & Algorithm Type & Field Constraint & Variable\\
 \hline
 \cite{Umans08}   &  \makecell{\small{$(d^m+N)^{1 + o(1)}\cdot $}\\ \small{$\poly(m, d,\log q)$}}   &algebraic& $p\leq d^{o(1)}$ & $m \leq d^{o(1)}$\\\hline
\cite{Kedlaya11}&  \makecell{\small{$(d^m+N)^{1 + o(1)}\cdot $}\\ \small{$\poly(m, d,\log q)$}}
& non-algebraic   & all finite fields & $ m\leq d^{o(1)}$\\\hline
\cite{BGKM21}    & \makecell{\small{$(d^m+N)^{1 + o(1)}\cdot $}\\ \small{$\poly(m, d,\log q)$}}
& algebraic &
 \makecell{$p\leq d^{o(1)}$,\\ 
 \small{$q \leq(\exp(\exp(\cdots(\exp (d))))$},\\ 
 where the height of this tower\\
  of exponentials is fixed} &  any $m$ \\
\hline
\makecell{This work\\ 
(Algorithm~\ref{algo:polynomial-evaluation-v1-ER})}  
&   \makecell{\small{$(d^m+N)^{1 + o(1)}\cdot $}\\ \small{$\poly(m, d,\log q)$}}
& non-algebraic & all finite fields & any $m$\\
\hline
\makecell{This work\\
(Algorithm~\ref{algo:polynomial-evaluation-v2-ext})} 
& \makecell{\small{$(d^m+N)^{1 + o(1)}\cdot $}\\ \small{$\poly(m, d,\log q)$}}
& non-algebraic & \makecell{\small{$q \leq(\exp(\exp(\cdots(\exp (d))))$},\\ 
where the height of this tower\\
 of exponentials is fixed} & any $m$\\
\hline
\end{tabular}
\label{tab:comparison}
\end{table}

\section{An Overview of the Proofs}\label{sec: proof overview}
In this section, we give an overview of the main ideas in our algorithms. At a high level, our algorithms rely on ideas from  three of the recent prior works on multivariate multipoint evaluation, namely that of Kedlaya and Umans \cite{Kedlaya11}, that of Bj\"{o}rklund, Kaski and Williams \cite{BKW19}, and a recent work of Bhargava, Ghosh, Kumar and Mohapatra \cite{BGKM21}. We start by giving a brief outline of these. 

We start with some necessary notation. Let $\F$ be a finite field and let $f \in \F[\vx]$ be an $m$-variate polynomial of degree less than $d$ in each variable, and let $\{\va_i: i \in [N]\}$ be a set of $N$ inputs in $\F^m$. Our goal is to evaluate $f$ on each $\va_i$.  For simplicity, we focus on the case when the underlying field $\F$ is a prime field, i.e. $\F = \F_p$ for some prime $p$. The case of extension fields is handled in a very similar manner, with a few technicalities.

A starting observation is that multivariate multipoint evaluation has a nearly linear time algorithm (over all fields) when the set of evaluation points forms a product set (see \autoref{lem:FFT-over-ring}), and more generally when the set of evaluation points is \emph{close} to a product set. At a high level, each of the algorithms in  \cite{Kedlaya11, BKW19, BGKM21} proceeds via a very efficient reduction from multivariate multipoint evaluation over an arbitrary set of points to multivariate multipoint evaluation over product sets. However, despite this common high level structure, the details of the reductions involved are fairly different in each of the three algorithms, thereby giving these algorithms their features, both desirable and undesirable. We now elaborate a bit more on these reductions. 

\subsection{The Algorithm of Kedlaya and Umans}\label{sec: KU overview}
To solve the problem efficiently over a finite field $\F$, Kedlaya and Umans \cite{Kedlaya11} first reduce an instance of the multivariate multipoint evaluation problem over $\F$ to an instance of the same problem over a ring of the form $\nbmodulo{r}$. Then they use their efficient algorithm for multivariate multipoint evaluation problem over $\nbmodulo{r}$ to solve it. Finally, from the evaluations over $\nbmodulo{r}$, they recover the original evaluations over $\F$. So, we first describe their algorithm over ring $\nbmodulo{r}$.

\paragraph*{The algorithm over $\nbmodulo{r}$: } In their algorithm, Kedlaya and Umans \cite[\defaultS4.2]{Kedlaya11} start by \emph{lifting} their problem instance over $\nbmodulo{r}$ to an instance over integers. They do this by just viewing $\nbmodulo{r}$ as the set of integers $\{0, 1, \ldots, r-1\}$ and this naturally maps a polynomial $f(\vx)$ over $\nbmodulo{r}$ to a polynomial $F(\vx)$ with coefficients in $\Z$. Similarly, this also gives a natural map from an input point $\va \in \modulo{r}^m$ to a point $\tilde{\va} \in \Z^m$. Clearly, for every $\va \in \modulo{r}^m$ and polynomial $f$, $f(\va) = F(\tilde{\va}) \mod r$. Thus, it suffices to solve this lifted instance over integers. Yet another property of this lifted instance is that the integer $F(\tilde{\va})$ is a non-negative integer of magnitude less than $M = d^m(r-1)^{dm}$ since each coefficient of $F$ and  each coordinate of $\tilde{\va}$ are in $\{0, 1, \ldots, r-1\}$,  and the total degree of $F$ is less than or equal to $(d-1)m$. Thus, to compute $F(\tilde{\va})$, it suffices to compute $F(\tilde{\va}) \mod M$. Kedlaya and Umans now proceed by finding distinct small primes $p_1, p_2, \ldots, p_k$ such that $\prod_{i \in [k]} p_i > M$, evaluating the polynomial $f_j(\vx) = F \mod p_j$ at the point $\vb_j = \tilde{\va} \mod p_j $ and \footnote{In other words, $f_j$ is obtained from $F$ by reducing each of its coefficients modulo $p_j$ and $\vb_j$ is obtained by reducing each of the coordinates of $\tilde{\va}$ modulo $p_j$.}  then combining the values $f_1(\vb_1), f_2(\vb_2), \ldots, f_{k}(\vb_k)$ using the Chinese Remainder Theorem. The correctness follows from the observation that for every $j \in [k]$, $f_j(\vb_j) = F(\tilde{\va}) \mod p_j$. The advantage of this \emph{multimodular} reduction is that if the primes $p_j$ are very small (for instance, if all these primes are close to $d$), then the set of evaluation points of interest, that were initially \emph{scattered} sparsely in $\F_p^n$ are now mapped to points that are packed densely in the space $\F_{p_j}^m$, which is a product set. Thus, we can use the simple multidimensional FFT to evaluate $f_j$ on all of $\F_{p_j}^m$ for every $j$, and then combine the outcome using the Chinese Remainder Theorem. For $m < d^{o(1)}$, this indeed gives a nearly linear time algorithm for multivariate multipoint evaluation. This constraint on the number of variables $m$ is due to a term of the form $(dm)^m$ in the final running time of the algorithm which is nearly linear in the input size only if $m$ is small. This $(dm)^m$ essentially appears because the product of primes $p_1, p_2, \ldots, p_k$ chosen in this reduction must exceed $M$, and hence, the largest of these primes $p_k$ must be $\Omega(\log M) = \Omega(dm \log r)$, and thus evaluating a polynomial $f_k$ on all of $\F_{p_k}^m$ requires at least $p_k^m = \Omega(d^mm^m)$ time. Recursive application of this process leads to smaller primes but the improved dependence is on the $\log r$ factor and this $(dm)^m$ factor continues to persist in the eventual bound on the running time. Thus, one approach towards a faster algorithm for multipoint evaluation over $\nbmodulo{r}$ would be to replace this step of evaluating $f_j$ on all of $\F_{p_j}^m$ in \cite{Kedlaya11} with a faster subroutine, in particular, something that runs in nearly linear time in the input size even for large $m$. 

Our first algorithm in this paper does precisely this. In order to obtain this gain, it crucially relies on ideas in an algorithm of  Bj\"{o}rklund, Kaski and Williams \cite{BKW19} which we discuss in \autoref{sec: BKW overview} and a very careful choice of primes to do Chinese Remaindering with, in the multimodular reduction discussed above. Together, these steps lead to an improvement in running time and give us an algorithm that runs in nearly linear time even when the number of variables is large. 

For our second algorithm, we introduce a slightly different modification in the framework of Kedlaya and Umans. Instead of working modulo small primes as in \cite{Kedlaya11}, which as discussed above, forces us to pick primes as large as $dm$, we work modulo powers of distinct primes in the multimodular reduction step. Thus, it seems conceivable that we can now work with much smaller primes than in the original algorithm, since instead of having the condition that the product of these primes is larger than $M$ as in \cite{Kedlaya11}, we now need that the product of powers of these primes is larger than $M$. However, we still need efficient algorithms for  multivariate multipoint evaluation over rings of the form $\nbmodulo{p^k}$ for small primes $p$ and large $k \in \N$. To handle this subproblem,  we extend the derivative-based techniques used in the algorithm of Bhargava et al. \cite{BGKM21} for fields of small characteristic so that they work over rings of the form $\nbmodulo{p^k}$ for small primes $p$ and large $k \in \N$. 

The advantage of this strategy over our first algorithm is that this gives us a completely elementary algorithm, and the disadvantage is that for this algorithm to run in nearly linear time, as desirable, the underlying ring $\nbmodulo{r}$ needs to be somewhat small. This issue also affects the original algorithm of Bhargava et al. \cite{BGKM21} and seems somewhat inherent to this style of an argument. 

\paragraph*{The algorithm over all finite fields:} 
We now give an outline of the algorithm of Kedlaya and Umans for extensions of prime fields.

Let $\F$ be the underlying finite field such that $|\F|=p^e$ for some prime $p$ and positive integer $e$. Then, we can assume that $\F$ is represented by $\F_p[z]/\ideal{E(z)}$ for some degree $e$ irreducible monic polynomial $E(z)$ over $\F_p$. Let $f(\vx)$ be the input polynomial over $\F$ with $m$ variables and degree less than $d$ in each variable and $\va_1,\va_2,\ldots,\va_N$ be the input points. Observe that each coefficient of $f(\vx)$ and each coordinate of $\va_i$'s are polynomial in $\F_p[z]$ of degree at most $e-1$. Like the previous case, Kedlaya and Umans \cite{Kedlaya11} lift $f(\vx)$ to a polynomial $F(\vx)\in\Z[z][\vx]$ and $\va_i$ to $\tilde\va_i\in \Z[z]^m$ by naturally identifying each element of $\F$ to a polynomial in $\Z[z]$ of degree at most $e-1$ and coefficients are in the set of integers $\{0,1,\ldots, p-1\}$. This reduces the problem of computing $f(\va_i)$ for all $i\in[N]$ to the problem of computing $F(\tilde\va_i)$ for all $i\in[N]$ since from $f(\va_i)$ is $F(\tilde\va_i)$ modulo $p$ and $E(z)$. 

Let $M=d^m(e(p-1))^{(d-1)m+1}+1$. One can observe that the coefficients of $F(\tilde\va_i)$, viewed as a polynomial in $z$ are all less than $M$. It follows from the fact that the evaluation of $F(\tilde \va_i)$ at $z=1$ is at most $M-1$. Therefore, we can recover $F(\tilde \va_i)$ by finding the $M$-base representation of the evaluation of $F(\tilde \va_i)$ at $z=M$. Also, note that the degree of $F(\tilde \va_i)$ in $z$ is at most $(e-1)dm$, hence the evaluation of $F(\tilde \va_i)$ at $z=M$ is less than $r=M^{(e-1)dm+1}$. Thus, computing $F(\tilde\va_i)$ modulo $z-M$ and $r$ is sufficient for computing $F(\tilde\va_i)$. This implies that they need to solve the following instance of the multivariate multipoint evaluation problem over the ring $\nbmodulo{r}$: the input polynomial is  $F(\vx)$ modulo $r$ and $z-M$, the evaluation points are $\tilde\va_i$ modulo $r$ and $z-M$. Now they invoke their multivariate multipoint evaluation over the ring $\nbmodulo{r}$ and get $F(\tilde \va_i)$ at $z=M$ for all $i\in[N]$.

\subsection{The Algorithm of Bj\"{o}rklund, Kaski and Williams}\label{sec: BKW overview}
In a nutshell, the algorithm of Bj\"{o}rklund et al. \cite{BKW19} proceeds via constructing a set $K \subseteq \F_p^m$ such that 
\begin{itemize}
    \item The size of $K$ is not too large and $K$ is (\emph{close} to) a product set. 
    \item For every $\va \in \F_p^m$, there is a curve $C_{\va}$ of low degree (in fact, a low degree univariate polynomial map) that passes through the point $\va$ and intersects the set $K$ on at least $p$ points~\footnote{This notion of a curve passing through a point here is slightly different to that in other related works like \cite{BGKM21}. However, for the sake of simplicity, we gloss over this technical detail right now.}.
    
\end{itemize}
These sets $K$ can be thought of as a natural higher degree analog of Kakeya sets over finite fields from discrete geometry. Indeed, Bj\"{o}rklund et al. refer to the set $K$ as high degree Kakeya sets, where the degree of the set is defined to be the maximum over the degrees of the curves $C_{\va}$ over all $\va \in \F_p^m$. 

Given such a Kakeya set $K$, Bj\"{o}rklund et al.proceed by evaluating $f$ on all points in $K$ fast, using the multidimensional FFT algorithm. This is the preprocessing phase of the algorithm. Then, for an arbitrary point $\va \in \F_p^m$, they compute $f(\va)$ by considering the univariate polynomial $R(y)$ obtained by taking the restriction $f$ on the curve $C_{\va}$. From the properties of the set $K$, we know the curve $C_{\va}$ intersects the set $K$ on at least $p$ points. Thus, if the degree of $R \leq \deg(f) \cdot \deg(C_{\va})$ is less than $p$, then we can recover the polynomial $R$ from the evaluations of $f$ on $K$ computed in the preprocessing step and using univariate polynomial interpolation. The quantitative bounds for this approach are therefore crucially determined by the size of the set $K$ and the degree of the curve $C_{\va}$.   

Bj\"{o}rklund et al. showed that for every $u \in \N$ such that $u+1$ divides $p-1$, there is a Kakeya set $K$ of degree $u$ of size at most $((p-1)/(u+1) + 1)^{m+1}$. This divisibility condition ensures the existence of a multiplicative subgroup of $\F_p^{*}$ of size $(p-1)/(u+1)$ and set $K$ is based on this subgroup. Thus, if $\tilde{d}$ denotes $(p-1)/(u+1)$, then we can evaluate the polynomial $f$ on the $K$ in time $\tilde{d}^m$, which is nearly linear in the input size if $\tilde{d} \leq d^{1 + o(1)}$. However, note that in this case, $u$ is around $p/\tilde{d}$, and hence, the degree of the restriction $R$ of $f$ on a curve of degree $u$ has total degree $udm = pm \cdot \frac{d}{\tilde{d}}$. Thus, if $pm \cdot \frac{d}{\tilde{d}} > p$,  we cannot hope to recover $R$ from its evaluations on just $p$ points. To address this issue, we  combine the above strategy in \cite{BKW19} with an idea in  \cite{BGKM21} where instead of evaluating just $f$ on $K$, we evaluate all its (Hasse) derivatives of order at most $m\cdot \frac{d}{\tilde{d}}$ on $K$ in the preprocessing phase. There are at most $\binom{m + m\cdot \frac{d}{\tilde{d}}}{m}$ such derivatives and this leads to an additional multiplicative factor of $\binom{m + m\cdot \frac{d}{\tilde{d}}}{m}$ in the final running time, but if $\tilde{d}$ is not too small compared to $d$, for instance, $\tilde{d} = \Theta(d)$, this binomial coefficient is at most $\exp(O(m))$ which is $d^{o(m)}$ for all growing $d$. Thus, with this stronger guarantee in the preprocessing step,  we are guaranteed to have higher multiplicity information available to us in the local computation step. So, we can now hope to uniquely recover a univariate polynomial of degree higher than $p$ from this information (via Hermite interpolation). However, since the degree of the univariates we have here is larger than $p$,  this Hermite interpolation step runs in time polynomially bounded in the underlying field size $p$ and not just polynomially bounded in $\log p$ as would have been desirable. 

To summarise, if there exists an $u \in \N$ such that $(p-1)/(u+1) = \tilde{d}$, where $\tilde{d}$ is close to $d$, e.g. $\tilde{d} = \Theta(d)$, then we have an algorithm for evaluating $m$-variate polynomials of degree less than $d$ in each variable on any $N$ points in $\F_p^m$ in time $\poly(p, d, m) \cdot (d^m + N)^{1 + o(1)}$. Thus, this is nearly linear time, when the field size $p$ is not too large. 

Having discussed these prior results, we are now ready to give an outline of our algorithms. We start with the first algorithm.
\subsection{The First Algorithm} \label{sec: algorithm 1 overview}
As discussed earlier in this section,  the plan for our algorithm is to somehow replace the multidimensional FFT step in the algorithm of Kedlaya and Umans \cite{Kedlaya11} (over rings of the form $\nbmodulo{r}$) with the Kakeya-set-based algorithm above over a field $\F_{p_j}$. However, in order to effectively use the Kakeya-set-based algorithm outlined in the previous section to obtain nearly linear time algorithms for multipoint evaluation, we need to ensure two properties. 
\begin{itemize}
    \item The underlying field size $p_j$ is small. For instance, we would need $p_j = (d^m + N)^{o(1)}$ for a nearly linear time algorithm.  
    \item There is a natural number $u$ such that $u+1$ divides $p_j-1$ and $(p_j-1)/(u+1)$ is an integer close to $d$. 
\end{itemize}
In fact, instead of the second condition here, it suffices if there is a small $t \in \N$ such that there exists a $u \in \N$ such that $u+1$ divides $p_j^t-1$ and  $(p_j^t-1)/(r+1) = d^{1 + o(1)}$, since we can always view the problem over $\F_{p}$ as a problem over an extension of $\F_p$. However, we need the degree of the extension to be small in order to get useful final quantitative bounds. 

The first condition about the primes $p_j$ being small does not appear too difficult to ensure in isolation and in particular, is also true for the algorithm of Kedlaya and Umans. However, the second divisibility condition seems trickier to guarantee even with the flexibility of working over low degree extensions of $\F_{p_j}$ as outlined earlier in this section. In particular, it is not clear to us if for every pair $d, p_j$, there always exists small $t$ such that $p_j^u-1$ has a divisor in the vicinity of $d$. 


Getting around these technical difficulties is the main technical content of our algorithm. 
In a nutshell, we proceed by following the multimodular reduction step of Kedlaya and Umans, but via a careful choice of primes $p_1, p_2, \ldots, p_k$ (as opposed to picking  a sufficiently large number of small primes as in \cite{Kedlaya11}). This careful choice preserves the fact that these primes are all small (at most $\poly(d, m, \log p)$) and additionally guarantees that the divisibility condition needed to invoke the Kakeya-set-based framework of \cite{BKW19}. More formally, we choose $p_1, p_2, \ldots, p_k$ so that they are all at most $\poly(d, m, \log p)$, their product exceeds $M = d^m(p-1)^{dm}$ and there exists a $\tilde{d} \in [0.8d, d]$ such that for every $j \in [k]$, $\tilde{d}$ divides $p_j - 1$. Thus, we can use the Kakeya-set-based framework outlined in \autoref{sec: BKW overview}, with the parameter $u_j$ to be set equal to $(p_j - 1)/\tilde{d} - 1$. This satisfies both the conditions highlighted earlier, and the final running time of this algorithm does indeed turn out to be nearly linear in the input size. The details can be found in  \autoref{sec: algo 1 prime fields}. Once we have this algorithm for multipoint evaluation over the rings of the form $\nbmodulo{r}$, we use exactly the same strategy as Kedlaya and Umans did to solve this problem over all finite fields. For details see \autoref{sec: algo 1 extension fields}.

Thus, if we can find distinct primes $p_1, p_2, \ldots, p_k$ with the properties outlined above, we would be done. However, it is not immediately clear how to do find such a set of numbers efficiently, or  whether such a collection of primes and the parameter $\tilde{d}$ should even exist. The appearance of the parameter $\tilde{d} = \Theta(d)$ is also slightly mysterious. For instance, it would be aesthetically nice if $\tilde{d}$ would have been  equal to $d$. Perhaps surprisingly, we do not know how to even show the existence of primes $p_1, p_2, \ldots, p_k$ satisfying the desired properties with $\tilde{d} = d$! We now outline our approach to finding such primes and the parameter $\tilde{d}$. However, for a start, let us attempt to do this with $\tilde{d} = d$ and try to understand the issues that arise. 

The intuition on showing the existence of such primes follows from the observation that if $d$ divides $p_j - 1$ for each $j \in [k]$ then, each of the primes $p_1, p_2, \ldots, p_k$ lies in the arithmetic progression (AP) $A_{d} = (1, 1 + {d}, 1 + 2{d}, \ldots)$. It follows from a classical theorem of Dirichlet (see Chapter $5$ in \cite{kedlaya15} for more details) that this arithmetic progression $A_{d}$ indeed contains an infinite number of primes for every ${d} \in \N$. Thus, if we take $k$ to be sufficiently large, then there exist primes $p_1, p_2, \ldots, p_k$ each congruent to $1$ modulo ${d}$ such that their product is greater than $M = d^m(p-1)^{dm}$. However, it is not enough for our application. We also need to show that these primes are not too large, e.g. each $p_i \leq \poly(d, m, \log p)$, and that they can be found efficiently. For this, it would be sufficient to show that not only does the arithmetic progression $A_{{d}}$ contains an infinite  number of primes, but the set of primes in $A_d$ is also a sufficiently dense subset of $A_d$. The prime number theorem gives such a statement for the progression $A_1$, i.e. for the set of natural numbers and here, a similar statement for arbitrary arithmetic progressions is needed. An unconditional bound on the density of primes in an arithmetic progression $A_d$ is given by the well-known Siegel-Walfisz theorem \cite{siegel1935classenzahl, walfisz1936additiven} which implies  a lower bound on the number of primes less than $x$ in the AP $A_d$ for all $x \geq 0$ with  $x > 2^{d^{\epsilon}}$ for any constant $\epsilon$. However, this estimate does not appear to be sufficient for us, since for the algorithm, we need the magnitude of these primes to be at most $\poly(d, m, \log p)$ and not exponentially growing in $d$, and it is not clear if such a guarantee can be obtained directly from this theorem. An improved lower bound on the density of primes in arithmetic progressions is known under the Generalized Riemann Hypothesis, and this would have been sufficient for our applications, except for the fact that the result would be conditional. For the unconditional result in this paper, we rely on the following theorem of Bombieri and Vinogradov, which gives an improved lower bound on the density of primes in an AP on average. For $x > 0, t \in \N$, let $\pi(x, t)$ be the number of primes less than $x$ in the AP starting at $1$ and with common difference $t$, $\pi(x)$ denote the number of primes less than $x$, and $\phi: \N \to \N$ be the Euler Totient function.  Various versions of this theorem can be found in literature, for instance, \cite{bombieri65, vinogradov1965density},  Theorem 18.1 in \cite{kedlaya15}. Here we rely on the bound in equation 1.1. in \cite{maynard}.
\begin{theorem}[Bombieri--Vinogradov]\label{them: BV}
For any fixed $a > 0$, there exist constants $c = c(a)$ and $b = b(a)$ such that for all sufficiently large $x > 0$,
\[
\sum_{t \leq d} \left |{\pi(x, t)} - \frac{\pi(x)}{\phi(t)} \right | \leq cx(\log x)^{-a} \, ,
\]
where $d \leq x^{1/2}(\log x)^{-b}$. 
\end{theorem}
Thus, if $x$ is sufficiently large compared to $d$, e.g. $x = d^3$, this theorem can be viewed as saying that on average (over $t \in \N, t \leq d$), an AP with common difference $t$ contains at least $\frac{\pi(x)}{\phi(t)} - cxd^{-1}(\log x)^{-a}$ primes less than $x$. Clearly, $\phi(t) \leq t \leq d$ and $\pi(x) = \Theta(x/\log x)$ by the prime number theorem. Thus, if we take $a > 1$, the number of primes less than $x$ is at least $\Omega(\pi(x)/d)$. For our final argument, we combine this average-case statement about the density of primes in an AP with a standard application of Markov's inequality to deduce that there exists a $\tilde{d} \in [0.8d, d]$ such that the AP with common difference $\tilde{d}$ has at least $\Omega(\pi(x)/\tilde{d})$ many primes less than $x$. By choosing $x$ to be a sufficiently large polynomial in $d, m, \log p$, we get precisely what we want: sufficiently many primes $p_1, p_2, \ldots, p_k$, each at most $\poly(d, m, \log p)$ in absolute value such that their product exceeds $M$ and they are all congruent to $1$ modulo $\tilde{d}$, for $\tilde{d} = \Theta(d)$. This application of Markov's inequality is precisely why we have to settle for working with the quantity $\tilde{d}$ and not $d$ itself. 

\subsection{The Second Algorithm}

In this section, we give a brief overview of our second algorithm. 
It implies that \autoref{thm: main informal} holds as long as the size of the finite field is bounded by $(\exp(\exp(\exp(\cdots (\exp(d)))))$, where the height of this tower of exponentials is fixed via an elementary algorithm. In particular, this algorithm does not rely on the Bombieri-Vinogradov theorem necessary for the first algorithm.

For simplicity, we first explain our algorithm over rings of the form $\nbmodulo{r}$, or $\nbmodulo{r^s}$ for some $s\leq m$. This covers the case of prime finite fields $\F_p$ by choosing $r=p$ and $s=1$.
After that, we briefly explain how to extend the algorithm to make it work over non-prime finite fields and certain extension rings of $\nbmodulo{r}$.

\paragraph*{The algorithm over $\nbmodulo{r}$:} 
Recall that Kedlaya and Umans \cite[\defaultS4.2]{Kedlaya11} use multimodular reduction together with the Chinese Remainder Theorem to reduce the multivariate multipoint evaluation problem over $\nbmodulo{r}$ to that  over $\F_{p_j}$ for a collection of small primes $p_j$.
As discussed in \autoref{sec: KU overview}, for the Chinese Remainder Theorem, the primes $p_j$ need to be chosen such that $\prod_{i\in [k]} p_i>M:=d^m (r-1)^{dm}$.
The problem here is that, as the primes $p_j$ are distinct, the largest prime would have order $O(\log M)=O(dm\log r)$. The $\log r$ factor can be further reduced by repeating the multimodular reduction. However, the $dm$ factor persists. As a consequence, the time complexity of the Kedlaya--Umans algorithm has a factor $(dm)^m$, which is nearly linear in $d^m$ only when $m=d^{o(1)}$.

In our algorithm, we introduce the new idea of using the \emph{prime powers} $p_j^m$ as the moduli for Chinese remaindering instead of the primes $p_j$.
That is, we compute the evaluations over the rings $\nbmodulo{p_j^m}$ and then combine them via Chinese Remainder Theorem to obtain the evaluations over the integers.
Assuming this can be done, then we only need to choose the primes $p_j$ such that $\prod_{i\in [k]} p_i^m>M$. So the largest prime may have order $O(\frac{1}{m}\log M)=O(d\log r)$, which is independent of $m$.

Now, to make this idea work, we need a fast algorithm for multivariate multipoint evaluation over $\nbmodulo{p_j^m}$, for small primes $p_j$. In particular, if we have an algorithm over $\nbmodulo{p_j^m}$ that runs in time $(p_j^m + N)^{1 + o(1)}$, then, overall, we have an algorithm that runs in time $(d^m(\log r)^m + N)^{1 + o(1)}$. Note that this has already enabled us to get rid of the $m^m$ factor in the running time as in \cite{Kedlaya11}. So, up to the factor of $(\log r)^m$  in the running time, we seem to have made some progress and we soon elaborate further on how to reduce this $(\log r)^m$ factor further. 
But first, we note that naively evaluating the polynomial at all points in $\modulo{p_j^m}^m$ would be extremely inefficient, as the size of $\modulo{p_j^m}^m$ is exponential in $m^2$. So, we need a significantly faster algorithm for multivariate multipoint evaluation over $\nbmodulo{p_j^m}$ to have any hope of making this strategy work. 

In their algorithm, Kedlaya and Umans \cite{Kedlaya11} deal with the $(\log r)^m$ factor by recursively applying the multimodular reduction a few times. So, to reduce the $(\log r)^m$ in the discussion above, we could also try to do something similar. We already see that one application of the reduction reduces the modulus $r$ to $p_j^m$ for a collection of primes $p_j$, where $\prod_{i\in [k]} p_i>d(r-1)^{d}$. Fix a prime $p_j$ and suppose we want to apply the multimodular reduction again. We may lift the instance over $\nbmodulo{p_j^m}$ to an instance over the integers, and then reduce it modulo $p'^m_i$ for a collection of primes $p'_i$.
The problem here is that, if we simply lift the evaluation points from $\modulo{p_j^m}^m$ to $\{0,1,\dots,p_j^m-1\}^m$, we would have an upper bound $M'=d^m(p_j^m-1)^{dm}$ for the evaluations over the integers, which is too large for us. The primes $p'_i$ would have to satisfy $\prod_i p'_i>M'^{1/m}=d(p_j^m-1)^{d}$, and then the order of the largest prime must depend (at least polynomially) on $m$.

We address the above two challenges, namely that of obtaining a fast multipoint evaluation algorithm over $\nbmodulo{p_j^m}$ that does not require evaluating on all of $\nbmodulo{p_j^m}^m$ and that of reducing the factor $(\log r)^m$ using the following observation: over $\nbmodulo{r^s}$, the evaluation of an $m$-variate polynomial $f(\vx)$ at a point $\var a\in \modulo{r^s}^m$ can be derived from the evaluations of the Hasse derivatives of $f(\vx)$ of sufficiently high order at another point $\var b\in \modulo{r^s}^m$, provided that the coordinates of $\var a-\var b$ are all multiples of $r$. Intuitively, this means if $\var a$ and $\var b$ are ``close enough'', then we can learn the evaluation of $f(\vx)$ at $\var a$ from the evaluations at $\var b$ of all the Hasse derivatives of $f$ of sufficiently high order. 

 Formally, for all $\ve\in\N^m$, let $\hpartial_{\ve}(f)\in\modulo{r^s}[\vx]$ be the Hasse derivative of $f(\vx)$ with respect to the monomial $\vx^\ve$.
For $\var a,\var b\in \modulo{r^s}^m$, we get from Taylor's expansion of $f(\vx)$ at $\var b$ that 
\[
f(\va)=\sum_{\ve\in\N^m}\hpartial_{\ve}(f)(\var b)(\va-\var b)^\ve. 
\]
Suppose the coordinates of $\var a-\var b$ are all multiples of $r$.
In this case, observe that $(\va-\var b)^\ve=0$ in $\nbmodulo{r^s}$ for all $\ve\in\N^m$ with $|\ve|_1\geq s$. Hence, 
\begin{equation}
\label{eqn:overview-2}
f(\va)=\sum_{\ve\in\N^m: |\ve|_1<s}\hpartial_{\ve}(f)(\var b)(\va-\var b)^\ve. 
\end{equation} 
So we may compute $f(\va)$ from the evaluations of Hasse derivatives $(\hpartial_{\ve}(f)(\var b))_{\ve\in\N^m:|\ve|_1<s}$.

We apply this idea to resolve the above two issues. First, in a base case of the recursive algorithm, instead of evaluating $f(\vx)$ at all points in $\modulo{p_j^m}^m$, we evaluate the Hasse derivatives $\hpartial_{\ve}(f)$ at the points in $S^m$ using a fast evaluation algorithm for product sets, where $S$ is the subset of $\nbmodulo{p_j^m}$ represented by $\{0,1,\dots,p_j-1\}$.
Note that for any $\var a\in \modulo{p_j^m}^m$, we may find $\var b\in S^m$ such that the coordinates of $\var a-\var b$ are multiples of $p_j$. Then $f(\va)$ can be computed from $\hpartial_{\ve}(f)(\var b)$ using \eqref{eqn:overview-2}.
The advantage of this is that the size of $S^m$ is only $p_j^m$, which is much smaller than the size $p_j^{m^2}$ of the whole set $\modulo{p_j^m}^m$.

Similarly, when applying the multimodular reduction over a ring $\nbmodulo{p_j^m}$, the idea above allows us to use a small yet non-exact lift of each evaluation point $\va_i$. Namely, suppose $\tilde{\va}_i\in\Z^m$ is the unique lift of $\va_i\in \modulo{p_j^m}^m$ with coordinates in $\{0,1,\dots,p_j^m-1\}$. We compute $\tilde{\va}'_i\in \{0,1,\dots,p_j-1\}^m$ whose coordinates are obtained by reducing the corresponding coordinates of $\tilde{\va_i}$ modulo $p_j$.
Then $\tilde{\va}'_i$ is a lift of some $\va'_i\in \modulo{p_j^m}^m$ such that the coordinates of $\va_i-\va'_i$ are all multiples of $p_j$.
We compute the evaluation $\hpartial_{\ve}(f)(\va'_i)$ at the point $\va'_i$ (instead of $\va$), and then $f(\va_i)$ can be computed from $\hpartial_{\ve}(f)(\va'_i)$ using \eqref{eqn:overview-2}. The advantage of evaluating at $\va'_i$ instead of $\va_i$ is that the coordinates of its lift $\tilde{\va}'_i$ are bounded by $p_j-1$ instead of $p_j^m-1$.
This translates into a better bound for the primes that we choose in multimodular reduction, thereby resolving the second issue.

Finally, at each level of the recursive algorithm, we need to evaluate not only $f(\vx)$, but also the Hasse derivatives $\hpartial_{\ve}(f)$ of order less than $m$. 
In addition, we need to solve the subproblem for each prime $p_j$.
This means the number of subproblems blows up by a factor of $2^{O(m)}\cdot O(d\log r)$ each time. However, as we assume the original $r$ (= the field size when $r$ is prime) is reasonably bounded in terms of $d$, it takes only a constant number of rounds to reduce $r$ to $d^{1+o(1)}$. So the total blow-up is reasonably controlled, and we obtain a nearly linear time algorithm when $d$ is sufficiently large. 
For details, see \autoref{sec:algo-2-rings}.

\paragraph{Comparison with the first algorithm: }
Compared to our first algorithm, which uses the ideas of generalized Kakeya sets and the Bombieri--Vinogradov theorem, our second algorithm uses a different idea, namely the Chinese Remainder Theorem with prime powers as the moduli. At a high level, this may be seen as an analogue of the ``method of multiplicities'' applied to the ring $\Z$ and polynomial rings over $\Z$. To see this, note that for a univariate polynomial $f(x)$ over a field, knowing the evaluations of all (Hasse) derivatives $f^{(i)}(x)$ of order $<s$ at a point $a$ is equivalent to knowing the remainder of $f$ modulo the power $\ideal{x-a}^s$. So from an ideal-theoretic point of view, the idea of applying the Chinese Remainder Theorem to learn an integer from its remainders modulo prime powers is analogous to applying Hermite interpolation to learn a univariate polynomial from the evaluations of its Hasse derivatives, the latter playing a crucial role in \cite{BGKM21}.

\paragraph*{The algorithm over finite fields (and extension rings of $\nbmodulo{r}$):} 
With further ideas, we extend our algorithm so that it works over arbitrary finite fields.
In fact, our algorithm works more generally over a ring $\modulo{r}[z]/(E(z))$, where $r\geq 2$ is an integer and $E(z)\in \modulo{r}[z]$ is a monic irreducible polynomial of degree $e\geq 1$.

Kedlaya and Umans \cite[\defaultS4.3]{Kedlaya11} described a reduction that reduces the problem of multivariate multipoint evaluation over $\modulo{r}[z]/(E(z))$ to that over $\nbmodulo{r'}$ for some integer $r'$.
Unfortunately, the integer $r'$ there is too large for us, being exponential in $m^2$. This is not a bottleneck in \cite{Kedlaya11}, as their algorithm over $\nbmodulo{r}$ already has a factor $m^m$ in its time complexity. However, it is a problem for us, so we cannot directly use the reduction in \cite{Kedlaya11}.

To achieve our claimed time complexity, we design a more efficient reduction, which reduces the evaluation problem to that over $\nbmodulo{r'^m}$, where $r'$ is independent of $m$. The basic idea is lifting the problem instance to an instance over $\Z[z]$, and then reducing it modulo $r'^m$ and $(z-j)^m$ for a small number of integers $j$. Here the idea of raising $z-j$ and $r'$ to their $m$-th powers helps us keep $r'$ small, and in particular, independent of $m$.
See \autoref{sec:algo-2-extn} for the details of the algorithm and a more thorough overview.

\section{Preliminaries}
Define $\N=\{0,1,\dots\}$, $\N^+=\{1,2,\dots\}$, $[n]=\{1,2,\dots,n\}$, and $\llbracket n \rrbracket=\{0,1,\dots,n-1\}$. The cardinality of a set $S$ is denoted by $|S|$. 

All rings in this paper are commutative rings with unity. 
For univariate polynomials $f(x), g(x)$ over a ring $R$ such that $g(x)$ is monic of positive degree, there exist unique $h(x),r(x)\in R[x]$ such that $f(x)=g(x)h(x)+r(x)$ and $\deg(r)<\deg(g)$ \cite[\defaultS IV.1, Theorem~1.1]{L02}. Define $f(x)\bmod g(x):=r(x)$, which can be computed using polynomially many $R$-operations via long division. 

By $\var x$ and $\var z$, we denote the variable tuples $(x_1,\ldots,x_m)$ and $(z_1,\ldots, z_m)$, respectively. For any $\var e=(e_1,\ldots, e_m)\in \N^m$, $\var x^\var e$ denotes the monomial $\prod_{i=1}^mx_i^{e_i}$. By $|\var e|_1$, we denote the sum $e_1+\cdots+e_m$. 

For every positive integer $k$, $k!$ denotes $\prod_{i=1}^ki$. For $k=0$, $k!$ is defined as $1$. For two non-negative integers $i$ and $k$ with $k\geq i$, $\binom{k}{i}$ denotes $\frac{k!}{i!(k-i)!}$. For $k<i$, $\binom{k}{i}=0$. For $\var a=(a_1,\ldots,a_m), \var b=(b_1,\ldots,b_m)\in\N^m$, $\binom{\var a}{\var b}=\prod_{i=1}^m\binom{a_i}{b_i}$.

\begin{proposition}
\label{prop:binomial-estimation}
For any two positive integers $i$ and $k$ with $k\geq i$, $$\binom{k}{i}\leq \left(\frac{ke}{i}\right)^i.$$
\end{proposition}
For a proof, see \cite[Chapter 1]{Jukna}. All logarithms in this paper are with respect to base $2$. For a non-negative integer $c$, $\log^{\circ c}(n)$ denotes the $c$-times composition of the logarithm function with itself. For example, $\log^{\circ 2}(n)=\log\log(n)$. We denote by $\log^{\star}(n)$ the smallest non-negative integer $c$ such that $\log^{\circ c}(n)\leq 1$. 

We need the following number-theoretic result.
\begin{lemma}[{\cite[Lemma~2.4]{Kedlaya11}}]
\label{lem:primes}
For all integers $N\geq 2$, the product of the primes $p\leq 16\log N$ is greater than $N$.
\end{lemma}
\subsection{Chinese Remainder Theorem}
For our algorithms, we crucially use the Chinese Remainder Theorem. For completeness, we formally state the version we use and refer to Chapter 10 of \cite{GG2013} for a proof. 
\begin{theorem}[Chinese Remainder Theorem]\label{thm:CRT}
Let $n_1, n_2, \ldots, n_{t}$ be pairwise relatively prime natural numbers greater than or equal to $2$ and let $u_1, u_2, \ldots, u_{t}$ be arbitrary natural numbers such that for every $i \in [t]$, $u_i \leq n_i - 1$. Then, there is a unique $v \in \N$ with $v < \prod_{i = 1}^{t}n_i$ such that for every $i \in [t]$, $v \equiv u_i \pmod{n_i}$. 

Moreover, there is a deterministic algorithm, that when given $n_1, n_2, \ldots, n_t$ and $u_1, u_2, \ldots, u_t$ as input, outputs $v$ in time at most $\poly(\sum_{i \in [t]} \log n_i)$, i.e., in time polynomial in the input size. 
\end{theorem}

\subsection{Hasse Derivatives}
In this section, we briefly discuss the notion of Hasse derivatives that plays a crucial role in our results. 
\begin{definition}[Hasse derivative]
\label{def:hasse-derivative}
Let $f(\var x)$ be an $m$-variate polynomial over a commutative ring $R$. Let $\var e=(e_1,\ldots, e_m)\in\N^m$. Then,  the Hasse derivative of $f$ with respect to the monomial $\var x^{\var e}$ is the coefficient of $\vecz^{\vece}$ in the polynomial $f(\vecx + \vecz) \in (R[\vecx])[\vecz]$. 
\end{definition}
\paragraph*{Notations.} Suppose that $f(\var x)$ is an $m$-variate polynomial over a commutative ring $R$. For $\va \in \N^m$, denote by $\hpartial_{\va}(f)$ the Hasse derivative of $f(\var x)$ with respect to the monomial $\var x^{\va}$. For any non-negative integer $k$,  define 
\[\hpartial^{\leq k}(f) := \left\{ \hpartial_{\var a}(f)\,\mid\,\var a\in\N^m \text{ s.t. }|\var a|_1\leq k \right\}
\]
and 
\[
\hpartial^{< k}(f):=\{\hpartial_{\var a}(f)\,\mid\,\var a\in\N^m \text{ s.t. } |\var a|_1< k\}.
\]

For a univariate polynomial $h(t)$ over $\F$ and a non-negative integer $k$, denote by $\der{h}{k}(t)$ the Hasse derivative of $h(t)$ with respect to the monomial $t^k$, that is, $\coeff_{z^k}(h(t+z))$.

Next, we mention a useful property of Hasse derivatives.  
\begin{proposition}
\label{lem:hasse-derivative-property}
Let $f(\var x)$ be an $m$-variate polynomial over a commutative ring $R$. Let $\var a,\var e\in\N^m$. Then, $$\hpartial_{\var e}(f)=\sum_{\var a\in\N^m}\binom{\var a}{\var e}\coeff_{\var x^{\var a}}(f)\var x^{\var a-\var e}.$$
\end{proposition}
For a proof, see, e.g., \cite[Appendix C]{F14}. 


The following lemma states that Hasse derivatives of polynomials can be computed efficiently. 

\begin{lemma}
\label{lem:computing-hasse-derivation}
Let $R$ be either a finite field or a ring of the form $\nbmodulo{r}$.
There exists an algorithm that given an $m$-variate polynomial $f(\vx)$ of individual degree less than $d$ over $R$ and $\ve\in\N^m$ with $|\ve|_1\leq dm$, computes $\hpartial_{\var e}(f)$ in time $O(d^m)\cdot\poly(m,d,\log |R|)$.
\end{lemma}

\begin{proof}
Let $S=\llbracket d\rrbracket^m$. For all $\va\in S$, let $c_\va$ denote the coefficient of monomial $\vx^\va$ in $f$. Then, from \autoref{lem:hasse-derivative-property}, we know that $$\hpartial_{\ve}(f)=\sum_{\va\in S}\binom{\va}{\ve}c_\va\vx^{\va-\ve}.$$
Without loss of generality, we can assume that each coordinate of $\ve$ is less than $d$. Otherwise, $\hpartial_{\var e}(f)$ is a zero polynomial. For any $\va\in S$, $\binom{\va}{\ve}$ can be computed in time $\poly(m,d,\log |R|)$. 
Thus, the time needed to compute $\hpartial_{\ve}(f)$ is $O(d^m)\cdot\poly(m,d,\log |R|)$. 
\end{proof}

\begin{remark*} For simplicity, we assume in \autoref{lem:computing-hasse-derivation} that $R$ is either a finite field or a finite ring of the form $\nbmodulo{r}$, as this will be sufficient for us.
The same assumption is made in \autoref{lem:hermitian-interpolation} and \autoref{lem:FFT-over-ring}, even though the lemmas and their proofs extend to general rings.
\end{remark*}


A useful additional ingredient in the proof of \autoref{thm: main informal} is the following lemma. Semantically, this is an explicit form of the chain rule of Hasse derivatives for the restriction of a multivariate polynomial to a curve of low degree.

\begin{lemma}
\label{lem:hasse-derivative}
Let $f(\var x)$ be an $m$-variate degree $d$ polynomial over a field $\F$, $\var g(t)=(g_1,\ldots, g_m)$ where $g_i\in\F[t]$, and $h(t)=f(\var g(t))$. For all $i\in[m]$, let $g_i(t+Z)=g_i(t)+Z\tilde g_i(t, Z)$ for some $\tilde{g_i} \in \F[t, Z]$. Let $\tilde{\var g}(t,Z)=(\tilde g_1,\ldots, \tilde g_m)$, and for all $\var e=(e_1,\ldots,e_m)\in\N^m$, $\tilde{\var g}_{\var e}=\prod_{i=1}^m\tilde g_i^{e_i}$. For any $\ell\in \N$, let $$h_{\ell}(t,Z)=\sum_{i=0}^\ell Z^i\sum_{\var e\in\N^m:|\var e|_1=i}\hpartial_{\var e}(f)(\var g(t))\cdot\tilde{\var g}_{\var e}(t,Z).$$ Then, for every $k\in\N$ with $k\leq \ell$, $\der{h}{k}(t)=\coeff_{Z^k}(h_\ell)$.
\end{lemma}

We now state a lemma that uses the above lemma for fast evaluation of all the Hasse derivatives of $h(t)=f(\var g(t))$ over a finite field $\F_q$. 
\begin{lemma} \label{lem:algo4}
Let $f(\vecx)$ be an $m$-variate, individual degree less than $d$ polynomial over a finite field $\F_q$ and $\var g(t)=(g_1, g_2, \ldots, g_m)$ where $g_i \in \F_q[t]$ with degree bounded by $r$. Then, given access to evaluations of $\hpartial^{\leq 2m}(f)$ on $\F_q$, there exists an algorithm that computes the evaluations of all $\leq 2m$ order Hasse derivatives of the polynomial $h(t)=f(\var g(t))$ at all points in $\F_q$ in time ${\Theta(1)}^m \cdot  \poly(q,r,d,m)$.
\end{lemma}

The proof of the above lemma (and its promised algorithm) follows directly from Algorithm~4 in \cite{BGKM21} and its correctness,  thus is skipped here. The only change is that Algorithm 4 looked at $\leq m$-th order Hasse derivatives, and here we are looking at $\leq 2m$-th order Hasse derivatives. It is an easy exercise to see that the analysis of the algorithm in \cite{BGKM21} extends as it is to this case.

\subsection{Hermite Interpolation} 
The following lemma gives a stronger version of univariate polynomial interpolation, known as Hermite interpolation. To interpolate a univariate polynomial of degree $d$, we need its evaluations at $d+1$ distinct points. However, for Hermite interpolation, the number of evaluation points can be less than $d$, provided that evaluations of Hasse derivatives of the polynomial are available up to a certain order.    
\begin{lemma}[Hermite interpolation]
\label{lem:hermitian-interpolation} 
Let $R$ be either a finite field or a ring of the form $\nbmodulo{r}$.
Let $f(x)$ be a univariate polynomial over $R$ and $e_1,\ldots, e_\ell$ be positive integers such that $d:=e_1+\cdots+e_\ell$ is greater than $\deg(f)$. Let $a_1, a_2,\ldots, a_\ell\in R$ such that for distinct $i,j\in [\ell]$, $a_i-a_j$ has multiplicative inverse in $R$. For all $i\in[\ell]$ and $j\in\llbracket e_j\rrbracket$, let $\beta_{ij}=\der{f}{j}(a_i)$. Then given $(a_i,\beta_{ij})$ for all $i\in[\ell]$ and $j\in\llbracket e_j\rrbracket$, $f(x)$ can be computed in time $\poly(d, \log |R|)$. Equivalently, given $(a_i, f(x)\bmod (x-a_i)^{e_i})$ for all $i\in [\ell]$, $f(x)$ can be computed in time $\poly(d, \log |R|)$.
\end{lemma}

\begin{proof}
Note $f(x)\bmod (x-a_i)^{e_i}=\sum_{j=0}^{e_j-1} \beta_{ij} (x-a_i)^j$ for $i\in [\ell]$. So $(a_i,\beta_{ij})$ and $(a_i, f(x)\bmod (x-a_i)^{e_i})$ can be computed from each other in time $\poly(d, \log |R|)$, and using either of them as the input is equivalent to using the other.

Next, we show that $f(x)$ can be computed in $\poly(d)$ $R$-operations given $f_i(x):=f(x)\bmod (x-a_i)^{e_i}$ for $i\in [\ell]$. When $R$ is a field, see \cite[\defaultS5.6]{GG2013} for a proof. We give a proof that works for general $R$. For $i\in [\ell]$, compute the following data.
First, compute
\[
r_i:=\prod_{j\in [\ell]\setminus\{i\}} (x-a_j)^{e_j}\bmod (x-a_i)=\prod_{j\in [\ell]\setminus\{i\}} (a_i-a_j)^{e_j},
\]
which is a unit in $R$ as each $a_i-a_j$ is a unit. Then compute $h_i(x)\in R[x]$ such that 
\[
\prod_{j\in [\ell]\setminus\{i\}} (x-a_j)^{e_j}=r_i-h_i(x)(x-a_i).
\]
Let $\lambda_i(x):=r_i^{-1}h_i(x)(x-a_i)$. 
Compute $\delta_i(x):=1-\lambda_i(x)^{e_i}$.
As $\lambda_i(x)$ is a multiple of $x-a_i$, we have $\delta_i(x)\equiv 1\pmod{(x-a_i)^{e_i}}$.
As $\delta_i(x)$ is a multiple of $1-\lambda_i(x)=r_i^{-1} \prod_{j\in [\ell]\setminus\{i\}} (x-a_j)^{e_j}$, we also have $\delta_i(x)\equiv 0\pmod{(x-a_j)^{e_j}}$ for $j\in [\ell]\setminus\{i\}$.
Finally, compute 
\[
g(x):=\sum_{i=1}^\ell \delta_i(x)f_i(x)~\bmod~ \prod_{i=1}^\ell (x-a_i)^{e_i}.
\]
Then $g(x)\equiv f_i(x)\equiv f(x)\pmod{(x-a_i)^{e_i}}$ for $i\in [\ell]$.
It remains to prove $g(x)=f(x)$. We know $g(x)-f(x)$ is a multiple of $(x-a_i)^{e_i}$ for $i\in [\ell]$. For distinct $i,j\in[\ell]$, the proof above constructs a multiple of $(x-a_i)^{e_i}$ whose remainder modulo $(x-a_j)^{e_j}$ is one. In particular, $(x-a_i)^{e_i}$ is multiplicatively invertible modulo $(x-a_j)^{e_j}$. So $(g(x)-f(x))/(x-a_i)^{e_i}$ is still a multiple of $(x-a_j)^{e_j}$ for all $j\in[\ell]\setminus\{i\}$. Repeating this argument shows that $g(x)-f(x)$ is a multiple of the degree-$d$ polynomial $\prod_{i=1}^\ell (x-a_i)^{e_i}$. As $\deg(g),\deg(f)<d$, we have $g(x)=f(x)$.
\end{proof}

\subsection{Fast Multivariate Multipoint Evaluation for Product Sets}

The following lemma states that multivariate multipoint evaluation can be solved very efficiently if the set of evaluation points is a product set. 
 \begin{lemma}
 \label{lem:FFT-over-ring} 
Let $R$ be either a finite field or a ring of the form $\nbmodulo{r}$.
 There exists an algorithm that given an $m$-variate polynomial $f(\vx)$ of individual degree less than $d$ over $R$ and a finite subset $S$ of $R$, outputs the evaluations $f(\va)$ for all $\va\in S^m$ in time $O(d^m+|S|^m) \cdot \poly(m,d,\log|R|)$.
 \end{lemma}
 \begin{proof}
 If $m=0$, $f\in R$ is just a scalar, and its evaluation at the only point in $S^0$ is $f$ itself. So just output $f$. 

 Now assume $m>0$. Compute $f_{a}:=f(x_1,\dots,x_{m-1}, a)$ for $a\in S$, which can be done in time $O(|S|d^m)\cdot\poly(m,d,\log|R|)$.\footnote{One can use FFT-based fast univariate multipoint evaluation \cite{GG2013} over $R[x_1,\dots,x_{m-1}]$ to compute all $f_a$ in time $d^{m-1}\cdot \tilde{O}(d+|S|)\cdot \poly(m, \log |R|)$, and the eventual time complexity would be $(d^{m-1}+|S|^{m-1})\cdot \tilde{O}(d+|S|)\cdot \poly(m,\log|R|)$. For us, the time complexity bound in \autoref{lem:FFT-over-ring} is good enough.}
 For each $a\in S$, recursively compute the evaluations $f_{a}(\va)$ for all $\va\in S^{m-1}$.
 Then output $(f_{a}(\va))_{\va\in S^{m-1}, a\in S}=(f(\va))_{\va\in S^{m}}$.

 Now we give an upper bound $T(m)$ for the time complexity of the above algorithm. We have $T(0)=O(1)$ and $T(m)=|S|\cdot T(m-1)+T'$ where $T':=O(|S|d^m)\cdot\poly(m,d,\log|R|)$. 
 
 When $|S|\leq d$, we have $T'=O(d^m)\cdot\poly(m,d,\log|R|)$. In this case, solving the recurrence relation using the fact $|S|\leq d$ yields $T(m)=O(d^m)\cdot\poly(m,d,\log|R|)$.
 When $|S|>d$, we have $T'=O(|S|^m)\cdot\poly(m,d,\log|R|)$, and solving the recurrence relation yields $T(m)=O(|S|^m)\cdot\poly(m,d,\log|R|)$.
 
 It follows that  $T(m)=O(d^m+|S|^m)\cdot\poly(m,d,\log|R|)$.
 \end{proof}


\section{The Necessary Building Blocks}\label{sec: building blocks}
In this section, we set up some of the necessary building blocks for our algorithm. 
\subsection{Primes in an Arithmetic Progression}
The first ingredient we need is the existence of \emph{sufficiently many} primes in the arithmetic progression $A_{d} = \{1, 1+d, 1 + 2d, \ldots \}$ that are not too large. When $d$ is small, and $x$ tends to infinity, a well-known result of Dirichlet (Theorem 5.5 in \cite{kedlaya15}) shows that the density of primes less than $x$ in the arithmetic progression $A_d$  tends to $\Theta(\frac{x}{\phi(d) \log x})$, where $\phi$ is the Euler totient function. However, for our application, we will need $x$ and $d$ to be \emph{close} to each other and hence it becomes important to carefully look at the error term in the prime counting function for the progression $A_d$. 

While we do not know how to show such a statement, we end up working with a weaker statement that turns out to be sufficient for our application. This weaker statement that we use  follows (immediately) from a deep result of Bombieri and Vinogradov that we state now. A more general statement can be found in Theorem 18.1 in \cite{kedlaya15}. But first, we need some notation. 
For any $x \geq 0$, we denote by $\pi(x)$ the number of primes less than or equal to $x$. For $x \geq 0 $ and $t \in \N$, we also use $\pi(x, t)$ to denote the number of primes less than or equal to $x$ in the arithmetic progression $A_t = \{1, 1+t, 1+ 2t, \ldots, \}$

We are now ready to state the theorem of Bombieri and Vinogradov that we use. Various versions of the theorem can be found in literature, for instance, \cite{bombieri65, vinogradov1965density},  Theorem 18.1 in \cite{kedlaya15}. Here we rely on the bound in Equation 1.1 in \cite{maynard}. 

\begin{theorem}[Bombieri-Vinogradov]
For any fixed $a > 0$, there exist constants $c = c(a)$ and $b = b(a)$ such that for all sufficiently large $x > 0$,
\[
\sum_{t \leq Q} \left |{\pi(x, t)} - \frac{\pi(x)}{\phi(t)} \right | \leq cx(\log x)^{-a} \, ,
\]
where $Q \leq x^{1/2}(\log x)^{-b}$.   
\end{theorem}

Semantically, \autoref{them: BV} says that on average (over $t \leq Q$), the quantity $\left |{\pi(x, 
t)} - \frac{\pi(x)}{\phi(t)} \right |$ is bounded by $(cx(\log x)^{-a})$. For our application, we would require a similar statement in the worst-case choice of $t$. This, however, is not known unconditionally when $t$ is large compared to $x$\footnote{More specifically, we would like $x$ and $t$ to be polynomially related to each other.} (which will turn out to be the case here), unless we assume  the Generalized Riemann Hypothesis. Thankfully, it turns out that we have some wriggle room, and we can in fact work with the average-case statement above (up to some small loss in the parameters). More formally, we need the following immediate consequence of \autoref{them: BV}. 
\begin{lemma}\label{lem: BV worst case}
For any fixed $a > 1$, there exist constants $c = c(a)$ and $b = b(a)$ such that for all sufficiently large $x > 0$, $Q \leq x^{1/2}(\log x)^{-b}$ and $\delta > 1$, there is a $t_0 \in \N$ with $Q(1-2/\delta) \leq t_0 \leq Q$ and
\[
{\pi(x, t_0)} \geq \frac{x}{4Q\log x}  \, .
\]

\end{lemma}
\begin{proof}
From \autoref{them: BV}, by dividing both sides by $Q$, we can view the summation as an expectation as $t$ varies uniformly in $[Q] = \{1, 2, \ldots, \lfloor Q \rfloor\}$. So, we get 
\[
\Ex_{t \in [Q]} \left[\left |{\pi(x, t)} - \frac{\pi(x)}{\phi(t)} \right |\right] \leq \frac{cx(\log x)^{-a}}{Q} \, .
\]
Now, by Markov's tail bound for the non-negative random variable 
$\left|{\pi(x, t)} - \frac{\pi(x)}{\phi(t)} \right|$, we get that for any $\delta > 1$, 
\[
\Pr_{t \in [Q]} \left[\left |{\pi(x, t)} - \frac{\pi(x)}{\phi(t)} \right | > \delta \cdot \frac{cx(\log x)^{-a}}{Q} \right] \leq 1/\delta \, .
\]
In particular, there exists an integer $t_0 \in [Q - 2Q/\delta, Q]$ such that  
\[
\left |{\pi(x, t_0)} - \frac{\pi(x)}{\phi(t_0)} \right | \leq  \delta \cdot \frac{cx(\log x)^{-a}}{Q}   \, .
\]
Or, in other words, 
\[
{\pi(x, t_0)} \geq \frac{\pi(x)}{\phi(t_0)} -  \delta \cdot \frac{cx(\log x)^{-a}}{Q}   \, .
\]
Now, since $x$ is sufficiently large, we have that $\pi(x) \geq (1 - o(1))x/\log x \geq x/2\log x$. So, \[
{\pi(x, t_0)} \geq \frac{x}{\log x} \left(\frac{1}{2\phi(t_0)} -  \delta \cdot \frac{c(\log x)^{1-a}}{Q}\right)   \, .
\]
Now, since $t_0 \leq Q$,  $\phi(t_0) \leq t_0 \leq Q$. So, we have 
\[
{\pi(x, t_0)} \geq \frac{x}{Q\log x} \left(\frac{1}{2} -  \delta \cdot {c(\log x)^{1-a}}\right)   \, .
\]
Finally, using the fact that $x$ is sufficiently large, and $c, \delta, a$ are constants with $a>1$, we get 
\[
{\pi(x, t_0)} \geq \frac{x}{4Q\log x}  \, .
\]
\end{proof}

We now state the following consequence of this lemma that will be directly useful for us in the Chinese Remaindering step of our algorithm. 
\begin{lemma}\label{lem: product of primes in an AP}
Let $D, M$ be natural numbers and let $D$ be sufficiently large. Then, there exists a natural number $\tilde{D} \in [0.8D, D]$ such that there are distinct primes $p_1, p_2, \ldots, p_k$ in the arithmetic progression $A_{\tilde{D}} = (1, 1 + \tilde{D}, 1 + \tilde{D}, \ldots, )$ with the following properties. 
\begin{enumerate}
    \item $k \leq  D^2 (\log M)^3$
    \item For every $i \in [k]$, $p_i \leq (D\log M)^3$  
    \item $\prod_{i = 1}^k p_i > M$
    
\end{enumerate}
Moreover, there is a deterministic algorithm that on input $D, M$ outputs $p_1,  \ldots, p_k, \tilde{D}$ in time $\poly(D, \log M)$.  
\end{lemma}
\begin{proof}
We invoke \autoref{lem: BV worst case} with the parameters $a$ set to an arbitrary positive constant greater than $1$, e.g. $a = 10$, $Q$ set to $D$, $x = (D \log M)^3$. Note that $D \leq  \sqrt{x}(\log x)^{-b}$ for $b = b(a)$ as given by \autoref{lem: BV worst case} for this choice of parameters. Let  the constant $\delta$ set to be $10$ (any arbitrary constant greater than $2$ works). Now, by \autoref{lem: BV worst case}, we get that there is a $\tilde{D} \in [0.8D, D]$ such that 
\[
{\pi(x, \tilde{D})} \geq \frac{x}{4D\log x}  \, .
\]
Let us consider the product of $k$ of these primes in the arithmetic progression $\{1, 1+\tilde{D}, 1 + 2\tilde{D}, \ldots, \}$ for any $k \leq \frac{x}{4D\log x}$. Clearly, each of these primes is at least $\tilde{D}$ (and hence $0.8D$), so their product is at least $(0.8 D)^k$. Thus, if $x$ is such that 
\[
\frac{x}{4D\log x} \geq \frac{\log M}{\log 0.8 D} \, ,
\]
then, we can always find sufficiently many distinct primes in the AP $A_{\tilde D}$ such that their product is at least $M$. This is true, for instance, for our choice of $x$ above.

For the moreover part, we try all possible values of $\tilde{D}$ in the range $[0.8D, D]$ and for each of these choices and $x = (D \log M)^3$, we check if the arithmetic progression $A_{\tilde{D}}$ has sufficiently many primes using the deterministic primality test of Agrawal, Kayal, and Saxena\cite{AKS04}. All these operations can be done deterministically in time $\poly(D, \log M)$. 

\end{proof}

\subsection{Explicit Kakeya Sets of Higher Degree}
We start with the definition of Kakeya sets of high degree. 
\begin{definition}[\cite{BKW19}]\label{def: kakeya}
Let $\F$ be a finite field and let $u, m \in \N$. A set $K \subseteq \F^m$ is said to be a Kakeya set of degree $u$ in $\F^m$ if there exist functions $g_0, g_1, \ldots, g_{u-1} : \F^m \to \F^m$ such that for every $\va \in \F^m$, the set of points 
\[
\{g_0(\va) + g_1(\va) \cdot \tau + \cdots + g_{u-1}(\va) \cdot \tau^{u-1} + \va \cdot \tau^u : \tau \in \F\}
\]
is a subset of $K$. 
\end{definition}
For ease of notation, we denote the curve
\[
\{g_0(\va) + g_1(\va) \cdot y + \cdots + g_{u-1}(\va) \cdot y^{u-1} + \va \cdot y^u : y \in \F\}
\]
of degree $u$ by $G_{\va}(y)$.

In their work \cite{BKW19}, Bj\"{o}rklund, Kaski and Williams gave an explicit construction of Kakeya sets of degree $u$ of non-trivially small size, provided that the degree $u$ and the field size $\F$ satisfy an appropriate divisibility condition. This construction will be crucial for our algorithm. 

\begin{theorem}[Explicit Kakeya sets of degree $u$ \cite{BKW19}]\label{thm: kakeya explicit}
Let $\F$ be a finite field of size $q$, and let $u \in \N$ be such that $u+1$ divides $q-1$. Then, for every $m \in \N$, there is a Kakeya set $K$ of degree $u$ in $\F^m$ of size at most $\left(\frac{q-1}{u+1} + 1 \right)^m$. 

Moreover, this set $K$ is a union of at most $q$ product sets in $\F^m$ and there is a deterministic algorithm that on input $u, m, \F$, outputs $K$ and the associated functions $g_0, g_1, \ldots, g_{u-1}$ in time $O(q|K|)$. 
\end{theorem}
\begin{proof}

We start by restating the Kakeya set of degree $u$ as described in \cite[Lemma 1]{BKW19}.  

\[
K:=\Bigg\{ \left( \Big( \frac{\alpha_1}{u+1} + \tau \Big)^{u+1} - \tau^{u+1},  \ldots, \Big( \frac{\alpha_m}{u+1} + \tau \Big)^{u+1} - \tau^{u+1}    \right) \Big| \, \forall \, \alpha_1, \alpha_2, \ldots, \alpha_m, \tau  \in \F_q  \Bigg\}.
\]

Observe that, $|K| \leq \left(\frac{q-1}{u+1} + 1 \right)^{m+1}$, as $|\{\beta^{u+1} : \beta \in \F_q \}| = \frac{q-1}{u+1} + 1$. Also, the associated functions $g_0, g_1, \ldots, g_{u-1}$  can be computed  in time $O(q|K|)$ by \cite[Lemma 1]{BKW19}. \\

We now show that $K$ is a union of at most $q$ product sets. 
For $\tau \in \F_q$, define  $S_{\tau}:= \{ \big( \frac{\alpha}{u+1} + \tau \big)^{u+1} - \tau^{u+1} \, | \, \forall \alpha  \in \F_q  \} \subseteq \F_q$. The proof concludes by observing that      $$K= \bigcup_{\tau \in \F_q} \underbrace{S_{\tau} \times S_{\tau} \times \cdots \times S_{\tau}}_{m \, \text{times}}.$$


\end{proof}
Using the property that the set $K$ in \autoref{thm: kakeya explicit} is a union of product sets and that for product sets we have nearly linear algorithms for multipoint evaluation using \autoref{lem:FFT-over-ring}, we get the following. 
\begin{lemma}\label{lem: evaluation on kakeya sets}
Let $\F$ be a finite field of size $q$, $u \in \N$ be such that $u + 1$ divides $q-1$, and $m \in \N$ be a natural number. Let $K$ be the Kakeya set of degree $u$ given by \autoref{thm: kakeya explicit} over $\F^m$ and let $f(\vx)$ be a polynomial of degree less than $d$ in each variable with coefficients in $\F$. 
 
 Then, there is a deterministic algorithm that takes as input the set $K$ and the coefficient vector of $f$ and outputs the evaluation of $f$ at every point in $K$ in time 
 \[
 O(|K| + d^m)\cdot  \poly(m, d, q) \, .
         \]
 \end{lemma}
\begin{proof}

Recall $K$ is a union of at most $q$ product sets, $K= \bigcup\limits_{\tau \in \F_q} \underbrace{S_{\tau} \times S_{\tau} \times \cdots \times S_{\tau}}_{ m \, \text{times}}$, here $S_{\tau}$ is as defined in proof of \autoref{thm: kakeya explicit}.  For each $\tau \in \F_q$, by using \autoref{lem:FFT-over-ring}, we can evaluate $f$ on $S_{\tau}$  in time $(d^m + |S_{\tau}|^m)\cdot \poly(m, d, \log |\F|)$. Thus, we can evaluate $f$ on every point in $K$ in time  $O(|K| + d^m)\cdot  \poly(m, d, q)  $.
\end{proof}

\subsection{Fast Multipoint Evaluation over Nice Finite Fields}

\begin{theorem}\label{thm: mme over nice fields}
Let $\F$ be a finite field of size $q$ and let $d, \Tilde{d}, m \in \N$ be such that $\tilde{d}\in [0.8d, d] $ and $\tilde{d}-1$ divides $q-1$.

Then, there is an algorithm that given a homogeneous  $m$-variate polynomial  in $\F[\vx]$ of degree less than $d$ in every variable  and a set of $N$ input points in $\F^m$, outputs the evaluation of this polynomial on these inputs in time 
\[
(d^m + N)\cdot \Theta(1)^m\cdot \poly(q, m, d) \, .
\]
\end{theorem}
\begin{proof}
Let $f$ be the input polynomial and let $\va_1, \va_2, \ldots, \va_N \in \F^m$ be the input points of interest. 

At a high level, the algorithm here is similar in structure to that in \cite{BGKM21}. We first evaluate the polynomial on an appropriate product set $\mathcal P$ in nearly linear time using \autoref{lem:FFT-over-ring} in the preprocessing phase. Next, in the local computation step, we look at the restriction of $f$ on a curve $C_{\va}$ \emph{through} any point $\va \in \F^m$ of interest. Based on the construction of the aforementioned product set $\mathcal P$, we will guarantee  that there is a curve $C_{\va}$ through $\va$ such that the intersection of $C_{\va}$ with the set $\mathcal P$ is sufficiently large, so that the univariate polynomial obtained by restricting $f$ to $C_{\va}$ can be uniquely \emph{decoded} using the evaluation of $f$ on $\mathcal P$. We then use this decoded polynomial to obtain $f(\va)$. 

Despite this high-level similarity, there are some technical differences between the algorithm here and that in \cite{BGKM21}. Primarily, these differences arise due to the fact that unlike the setting in \cite{BGKM21}, we are no longer working over fields of small characteristic. So, the construction of the set $\mathcal P$ is different here and is based on the ideas in \cite{BKW19}. We now specify the details, starting with the description of the algorithm. 

\paragraph{The algorithm.}
\begin{enumerate}
    \item From the coefficient vector of $f$, compute each of Hasse derivatives of $f$ of order at most $2m$. 
    \item Using \autoref{thm: kakeya explicit}, we construct a Kakeya set $K$ of order $u = (q-1)/(\tilde{d}-1) - 1$. As is necessary, $u+1$ divides $q-1$. Note that, $|K| \leq \Tilde{d}^{(m+1)}$.
    \item For every Hasse derivative $\tilde{f}$ of $f$ order at most $2m$, evaluate $\tilde{f}$ on $K$ using \autoref{lem: evaluation on kakeya sets}. 
    \item For every $i \in [N]$: \begin{enumerate}
        \item We consider the univariate polynomial $R_i(y)$ obtained by the restriction of $f$ on the curve $G_{\va_i}(y)$. This is a univariate polynomial of degree at most $(d-1)m\cdot (q-1)/(\tilde{d}-1) <  2m(q-1)$. Using \autoref{lem:algo4}, compute the evaluation of $R_i(y)$ and all its $\leq 2m$ order Hasse derivatives on $\F$. 
        \item Since degree of $R_i$ is less than $2m(q-1)$, and we have the evaluation of $R_i$ and all its derivatives of order at most $2m$ on $q$ points, we can recover $R_i$ uniquely from this information. In particular, we use \autoref{lem:hermitian-interpolation} to recover $R_i(y)$. 
    \item We output $f(\va_i)$ to be equal to the leading coefficient of $R_i(y)$. 

    \end{enumerate}

\end{enumerate}

\paragraph{Running time.}
There are a total of $\binom{m+2m}{m}$ Hasse derivatives of order at most $2m$ of $f$. By \autoref{lem:computing-hasse-derivation}, each Hasse derivative takes time $d^m \cdot \poly(m,d,\log q)$. Thus, the total time is bounded by $\Theta(d)^m \cdot \poly(m,d,\log q)$. By \autoref{thm: kakeya explicit}, we can compute the explicit Kakeya set of degree $u:=\frac{q-1}{\tilde{d}-1}-1$ in time $O(|K|q) = \Theta(d)^m \cdot q $.    
For the time complexity of Step 3, by  \autoref{lem: evaluation on kakeya sets},  we can evaluate $f$ and all its Hasse derivatives of $f$ of order at most $2m$ in time $$\binom{m+2m}{m} \cdot \left( \Tilde{d}^{m+1} + d^m  \right) \cdot \poly(m,d,q) = \Theta{(d)}^m \cdot  \poly(m,d,q).$$ For the loop in Step 4, the time complexity is bounded by ${\Theta(1)}^m \cdot N \cdot \poly(d,q)$ by \autoref{lem:algo4}. By Hermite interpolation  (\autoref{lem:hermitian-interpolation}), we can recover $R_i(y)$ for each $i$ in time $\poly(m,q)$. Thus, the total time complexity is bounded by $ (d^m+ N)\cdot \Theta(1)^m \cdot \poly(m,d,q)$.


\paragraph{Correctness.}
Note that, for any \textit{homogenous} polynomial $f (\vecx) $, the leading coefficient of $R_i(y)=f(g_0(\va_i) + g_1(\va_i) \cdot y + \cdots + g_{r-1}(\va_i) \cdot y^{r-1} + \va \cdot y^r)$ is  $f(\va_i)$.
Thus, given the polynomial $R_i(y)$, then we get $f(\va_i)$ by directly reading off  its leading coefficient. Since the degree of $R_i(y)$ is $\leq 2m(q-1)$, via Hermite interpolation (\autoref{lem:hermitian-interpolation}) it suffices to know the evaluations of $\leq 2m$ order Hasse derivatives of  $R_i(y)$ on $\F_q$. 
Finally, note that, we have access to evaluations of $\leq 2m$ order partial derivatives of  $R_i(y)$ for all $i \in [N]$ because of \autoref{lem:algo4} and \autoref{lem: evaluation on kakeya sets}. This concludes the proof.
\end{proof}

\section{The First Algorithm over Rings of the Form $\nbmodulo{r}$}
\label{sec: algo 1 prime fields}
With the necessary background in place, we are now ready to describe our first algorithm for fast multivariate multipoint evaluation over rings of the form $\nbmodulo{r}$. This already handles the case of prime fields, and contains most of our main ideas. 
 Later, in  \autoref{sec: algo 1 extension fields}, we discuss the case of extension rings which will complete the proof of  
\autoref{thm: main informal}. Also, the algorithm in \autoref{sec: algo 1 extension fields} crucially relies on the algorithm for the $\nbmodulo{r}$ case and an idea of Kedlaya and Umans \cite{Kedlaya11}.   

\subsection{The Description of the Algorithm}

{
\centering
\begin{minipage}{\algwidth}
\begin{algorithm}[H]
\caption{The First Algorithm over Rings of the Form $\nbmodulo{r}$}
\label{algo:polynomial-evaluation-v1}
Algorithm \textsc{MME-A}$(f(\vx),\va_1,\va_2,\ldots, \va_N, r)$

\medskip
\noindent where $f$ is an $m$-variate \emph{homogeneous} polynomial over $\nbmodulo{r}$ of individual degree less than $d$ and $\va_1, \va_2, \ldots, \va_N \in \modulo{r}^m$ are evaluation points.
\begin{enumerate} 
    \item Let $F \in \Z[\vx]$ be the $m$-variate homogeneous polynomial of individual degree less than $d$ obtained from $f$ by replacing each of its coefficients with its natural lift in the set $\llbracket r\rrbracket$ of integers. 
    \item For every $i \in [N]$, let $\tilde{\va}_i \in \llbracket r\rrbracket^m$ be the lift of $\va_i \in \modulo{r}^m$ to the integers. 
    \item Let $M = d^m r^{dm}$. We invoke \autoref{lem: product of primes in an AP} with parameters $d$ and $M$ and obtain a natural number $\tilde{d} \in [0.8d, d]$ and primes $p_1, p_2, \ldots, p_k$ where $k \leq d^2(\log M)^3$, each $p_i \leq d^3(\log M)^3$ and is congruent to $1$ modulo $\tilde{d}$, and $\prod_{i \in [k]}p_i > M$.  
    \item For $j \in [k]$, let $f_j(\vx) \in \F_{p_i}[\vx]$ be the $m$-variate homogeneous polynomial of individual degree less than $d$ obtained by reducing each of the coefficients of $F$ modulo the prime $p_j$. Similarly, for every $i \in [N]$, let $\va_{i, j} \in \F_{p_j}^m$ be obtained by reducing each of the coordinates of $\tilde{\va}_i$ modulo $p_j$.   
    \item For every $j \in [k]$, invoke the algorithm in \autoref{thm: mme over nice fields} for the polynomial $f_j$, input points $\{\va_{i, j} : i \in [N]\}$ and parameters $d, \tilde{d}$ as above, and get $f_j(\va_{i,j})$ for all $j\in[k]$ and $i\in[N]$. Note that each $f_j$ is a homogeneous polynomial, and from the guarantees of \autoref{lem: product of primes in an AP}, $\tilde{d}$ is in the range  $[0.8d, d]$ and $\tilde{d} -1 $ divides $p_i - 1$ as needed by \autoref{thm: mme over nice fields}. 
    \item For every $i \in [N]$, use the Chinese Remainder Theorem (\autoref{thm:CRT}) to compute  $F(\tilde{\va}_i)$ from $\{f_j(\va_{i, j}) : j \in [k]\}$. 
    \item For every $i \in [N]$, output $f(\va_i) = F(\tilde{\va}_i) \mod r$.
    \end{enumerate}
\end{algorithm}
\end{minipage}
}

\subsection{The Correctness of Algorithm \ref{algo:polynomial-evaluation-v1}}
The correctness of the algorithm essentially follows from the correctness of the building blocks in \autoref{sec: building blocks} and preliminary notions like \autoref{thm:CRT}. We now elaborate on some of the details. 

In its first two steps, Algorithm \ref{algo:polynomial-evaluation-v1} lifts the problem of multipoint evaluation given over the ring $\nbmodulo{r}$ to an instance over $\Z$ by naturally identifying the elements of $\nbmodulo{r}$ with the set $\llbracket r\rrbracket$ of integers. This lifted instance is given by the polynomial $F \in \Z[\vx]$ and the points $\{\tilde{\va}_i : i \in [N] \}$. Thus, to correctly solve the original problem over $\nbmodulo{r}$, it suffices to correctly solve this lift over $\Z$ and then reduce the outputs modulo $r$, since for every $i \in [N]$, $f(\va_i) = F(\tilde{\va}_i) \mod r$. Hence, it suffices to argue that the computation of $F(\tilde{\va}_i)$ is correct for every $i \in [N]$. 

Since $F$ is an $m$-variate polynomial of degree less than $d$ in each variable with every coefficient being in the set $\llbracket r\rrbracket$, and for every $i \in [N]$, each coordinate of $\tilde{\va}_i$ is also in the set $\llbracket r\rrbracket$, we get that $F(\tilde{\va}_i)$ is a natural number of absolute value at most $d^m(r-1)^{dm}$. Thus, for $M = d^mr^{dm} > d^m(r-1)^{dm}$, it suffices to compute $F(\tilde{\va}_i)$ modulo $M$. 

This computation modulo $M$ is done by working modulo distinct primes $p_1, p_2, \ldots, p_k$, each not too large, such that their product exceeds $M$. Moreover, these primes are carefully chosen using \autoref{lem: product of primes in an AP} which additionally ensures that there is a $\tilde{d} \in [0.8d, d]$ such that each of these primes is in the arithmetic progression $(1, 1 + \tilde{d}, 1 + 2\tilde{d}, \ldots )$. 

Given this additional structure on the primes, we next use \autoref{thm: mme over nice fields} to  solve the problem of evaluating the polynomial $f_j = F \mod p_j$ on inputs $\{\va_{i, j} : i \in [N]\}$, where $\va_{i, j} = \tilde{\va}_i \mod p_j$ over the field $\F_{p_j}$. Since $f_j$ is homogeneous and the divisibility condition needed in the hypothesis of \autoref{thm: mme over nice fields} holds, we have that this step works correctly and within the desired time guarantees.  In other words, for every $i \in [N]$, the inputs to the Chinese Remainder Step (Step $6$) of the algorithm are all correct. Thus, each of the evaluations of $F$ and hence, those of $f$ output by the algorithm are correct.

\subsection{The Time Complexity of Algorithm \ref{algo:polynomial-evaluation-v1}}
\label{subsec:time-complexity-algo1}
To bound the time complexity, we bound the time complexity of each of the steps of the algorithm. 

The first two steps of the algorithm essentially require no additional computation beyond a reading of the input. We just semantically re-interpret the coefficients of $f$ and the coordinates of $\va_i$ to be over the integers. Thus, these can be done in time $(d^m + N)\cdot \poly(\log r, m, d)$. From the choice of $M$ in Step $3$ and the guarantees in \autoref{lem: product of primes in an AP}, it follows that the time complexity of Step $3$ is at most  $\poly(\log M, d) = \poly(d, m, \log r)$. Step $4$ again takes at most $(d^m + N)\cdot\poly(d, m, \log r)$ time since each of the primes $p_i$ has absolute value at most $\poly(d, m, \log r)$ from \autoref{lem: product of primes in an AP}. For every $j \in [k]$, it follows from \autoref{thm: mme over nice fields} that multipoint evaluation of the polynomial $f_j$ on inputs $\{\va_{i, j} : i \in [N]\} \subseteq \F_{p_j}^m$ takes at most $(d^m + N)\cdot \Theta(1)^m \cdot \poly(d, m, p_j) = (d^m + N)\cdot\Theta(1)^m \cdot \poly(d, m, \log r)$ from the bound on the magnitude of $p_j$. However, we note that to use \autoref{thm: mme over nice fields}, we need the inputs in this function call to satisfy the  divisibility condition in the hypothesis of the theorem, namely that $\tilde{d} - 1$ divides $p_j -1$. But we have already argued this while arguing the correctness of the algorithm.

Step $6$ is a straightforward application of the Chinese Remainder Theorem and using \autoref{thm:CRT}, we can do this in time at most $N \cdot \poly(k, \max_j (\log p_j)) \leq N\cdot \poly(d, m, \log r)$. Finally, the last step is just a sequence of $N$ integer divisions involving numbers of bit complexity at most $\log M$, and hence can be implemented in time $N\cdot \poly(\log M) = N \cdot \poly(d, m, \log r)$. 

Thus, the overall time complexity of the algorithm is at most $\left(d^m + N \right)\cdot\Theta(1)^m \cdot \poly(d, m, \log r)$. 

We summarize the above discussion in the following theorem.
\begin{theorem}
\label{thm:first-algo-ring}
Let $f(\vx)$ be a homogeneous $m$-variate polynomial over $\nbmodulo{r}$ of individual degree less than $d$. Let $\va_1,\va_2,\ldots, \va_N$ be $N$ points from $\modulo{r}^m$. Then, given $(f, \va_1,\va_2,\ldots, \va_N,r)$ as the input to Algorithm \ref{algo:polynomial-evaluation-v1}, it computes $f(\va_i)$ for all $i\in[N]$ in time $$(d^m+N)\cdot\Theta(1)^m\cdot \poly(m,d,\log r).$$
\end{theorem}

\section{The First Algorithm over Extension Rings}
\label{sec: algo 1 extension fields}
In this section, we extend the fast multivariate multipoint evaluation algorithm discussed in the previous section to over extension rings of the form $\modulo{r}[z]/\ideal{E(z)}$ for a polynomial $E(z) \in \nbmodulo{r}[z]$, and in particular over all finite fields. This will prove our 
\autoref{thm: main informal}. Here, the underlying ring is of the form $\modulo{r}[z]/\ideal{E(z)}$ where $E(z)$ is a monic polynomial of degree at most $e-1$. The idea for this extension is essentially the same as that used by  Kedlaya and Umans in \cite[\defaultS4.3]{Kedlaya11} to extend their algorithm over rings of the form $\nbmodulo{r}$ to other extension rings. They essentially reduce an instance of multivariate multipoint evaluation over extension rings to an instance over rings of the form $\nbmodulo{r'}$. Then, they invoke their multivariate multipoint evaluation algorithm over $\nbmodulo{r'}$. For our proof, we use our Algorithm \ref{algo:polynomial-evaluation-v1} for solving the problem over $\nbmodulo{r'}$. The improved dependence on the number of variables for our algorithm thus just follows from the improved dependence on the number of variables in the complexity of Algorithm \autoref{algo:polynomial-evaluation-v1}. The rest of the steps are the same as \cite{Kedlaya11}. We now describe the steps of the algorithm.  

\subsection{The Description of the Algorithm}

{
\centering
\begin{minipage}{\algwidth}
\begin{algorithm}[H]
\caption{The First Algorithm over Extension Rings}
\label{algo:polynomial-evaluation-v1-ER}
Algorithm \textsc{MME-for-extension-rings-A}$(f(\vx),\va_1,\ldots,\va_N)$

\medskip
\noindent where $R$ is the underlying ring represented as $\modulo{r}[z]/\ideal{E(z)}$ such that $E(z)$ is a degree $e$ monic polynomial over $\nbmodulo{r}$, $f(\vx)$ is an $m$-variate \emph{homogeneous} polynomial over $R$ of individual degree less than $d$, and $\va_1,\va_2,\ldots, \va_N$ are points in $R^m$.

\medskip
\noindent Let $M:=d^m(e(r-1))^{(d-1)m+1}+1$ and $r':=M^{(e-1)dm+1}$.
\begin{enumerate}
\item Let $F(\vx)\in\Z[z][\vx]$ be the $m$-variate homogeneous polynomial of individual degree less than $d$ obtained by replacing each coefficient of $f(\vx)$ (which is a polynomial in $\nbmodulo{r}[z]$) with its natural lift to a  polynomial in $z$ over the integers by identifying $\nbmodulo{r}$ with the set of integers $\llbracket r \rrbracket$. 
Similarly, for every $i\in[N]$, let $\tilde \va_i\in\Z[z]^m$ be the lift of the point $\va_i\in R^m$.

\item Let $\overline{f}(\vx)$ be the polynomial computed from $F(\vx)$ by reducing its coefficients (which are elements of $\Z[z]$) modulo  $z-M$ and $r'$, i.e., for each of these polynomials in $\Z[z]$, we first set $z = M$ and then reduce the result modulo $r'$. Similarly, for all $i\in[N]$, $\overline\va_i$ is the point obtained from $\tilde\va_i$ by reducing each of its coordinates modulo $z-M$ and $r'$. 

Note that from the choice of $r'$, going modulo $r'$ does not change anything computationally, but formally reduces the problem to an instance of multivariate multipoint evaluation over $\nbmodulo{r'}$. 

\item Call the function \textsc{MME-A}$(\overline f,\overline\va_1,\overline\va_2,\ldots,\overline\va_N, r'))$ in Algorithm \ref{algo:polynomial-evaluation-v1} and get $\overline b_i=\overline f(\overline\va_i)$ for all $i\in[N]$.

\item For all $i \in [N]$, from $\overline b_i$ compute the unique polynomial $Q_i(z)$ in $\Z[z]$ of degree at most $(e-1)dm$ and coefficients are in $\llbracket M\rrbracket$ such that $Q_i(M)$ is congruent to $\overline{b_i}$ modulo $r'$. In other words, we compute  base $M$ representation of the natural number $\overline b_i$. Finally, we compute $Q_i(z)$ modulo $r$ and $E(z)$ and get $f(\va_i)$ for all $i\in[N]$. 
\end{enumerate}
\end{algorithm}
\end{minipage}
}

\subsection{The Correctness of Algorithm \ref{algo:polynomial-evaluation-v1-ER}}
\label{sec:correctness-poly-eval-1-ER}
We show that Algorithm \ref{algo:polynomial-evaluation-v1-ER} successfully computes $f(\va_i)$ for all $i\in[N]$. From the first step of the above algorithm, we observe that for all $i\in[N]$, $f(\va_i)$ equals  $F(\tilde \va_i)$ modulo $r$ and $E(z)$. Therefore, at the end of step 1, Algorithm \ref{algo:polynomial-evaluation-v1-ER} reduces the problem of computing $f(\va_i)$ for all $i\in[N]$ over the extension ring to the problem of computing $F(\tilde\va_i)$ for all $i\in[N]$ over integers. The next natural step is to solve this problem of computing $F(\tilde\va_i)$ for all $i\in[N]$ using Algorithm \ref{algo:polynomial-evaluation-v1}. To this end, we further reduce this problem of computing $F(\tilde \va_i)$ for $i\in[N]$ to an instance of multivariate multipoint evaluation over rings of the form $\nbmodulo{r'}$. 

From the construction of $F$, we know that the coefficients of $F(\vx)$ and the coordinates of $\tilde\va_i$ are all polynomials in $z$ of degree at most $e-1$, and  the coefficients of these polynomials are integers in the set $\llbracket r\rrbracket$. Thus, $F(\tilde\va_i)$ is a polynomial in $z$ of degree at most $(e-1)dm$ and each of its coefficients is a non-negative integer of absolute value at most $d^m \cdot (e(r-1))^{(d-1)m}\cdot (r-1) \leq d^m \cdot (e(r-1))^{(d-1)m+1}$, which by our choice of $M$ is at most $M-1$. Recall that an arbitrary polynomial $P \in \Z[z]$ with non-negative integer coefficients of absolute value less than $M$ can be uniquely (and efficiently) recovered given its evaluation at $M$: just construct the base $M$ representation of $P(M)$ (or equivalently $P(z) \mod (z-M)$), and read off the digits of such a representation. Based on this, we set ourselves the goal of computing $F(\tilde\va_i)$ by computing $F(\tilde\va_i) \mod (z-M)$. From the choice of $r'$, $r'$ is strictly larger than the integer $F(\tilde\va_i) \mod (z-M)$ and hence it suffices to compute $F(\tilde\va_i) \mod (z-M)$ while working modulo $r'$. Note that this is equal to $\overline f(\overline\va_i)$ in our notation. This reduction is done in step $2$ of the above algorithm and reduces the original problem to an instance of multipoint evaluation over $\nbmodulo{r'}$. This computation, in turn is done in step $3$ of the algorithm using Algorithm \autoref{algo:polynomial-evaluation-v1}. The output of this step is $\overline b_i=\overline f(\overline\va_i)$ for all $i\in[N]$. From $\overline b_i$, we get $F(\tilde\va_i)$ (which is the same as $Q_i(z)$ in the algorithm) by computing the representation of the number $\overline b_i$ with respect to the base $M$. From this, we get $f(\va_i)$ easily by reducing $Q_i(z)$ modulo $r$ and $E(z)$. This completes the proof of correctness of the algorithm. 



\subsection{The Time Complexity of Algorithm \ref{algo:polynomial-evaluation-v1-ER}}
In step 1 of Algorithm \ref{algo:polynomial-evaluation-v1-ER}, the cost of lifting the polynomial $f(\vx)$ and the points $\va_1,\va_2,\ldots,\va_N$ to over integers is $O((d^m+mN)e\log r)$. Step 2 of the algorithm computes the polynomial $\overline f(\vx)$ which is $F(\vx)$ modulo $z-M$ and $r'$. Computing $F(\vx)$ modulo $z-M$ is the same as evaluating each of its coefficients at $z=M$. Since each coefficient of $F(\vx)$ is a polynomial in $z$ of degree at most $e-1$ and coefficients are in the set of integers $\llbracket r\rrbracket$, evaluating each coefficient of $F(\vx)$ at $z=M$ can be done in time $e\cdot \poly( \log r,\log M)$ which is at most $\poly(d,m,e,\log r)$. Also, note that computationally reduction modulo $r'$  does not involve any actual computation since $r'$ is very large and hence the cost of computing $\overline f(\vx)$ is $d^m\cdot \poly(d,m,e,\log r)$. Similarly, computing $\overline\va_i$ for all $i\in[N]$ can be done in time  $N\cdot\poly(d,m,e,\log r)$. Thus, the cost of the step 2 of Algorithm \ref{algo:polynomial-evaluation-v1-ER} is $(d^m+N)\cdot \poly(d,m,e,\log r)$. Step 3 of the algorithm invokes the function \textsc{MME-A}$(\overline f,\overline \va_1,\overline\va_2,\ldots,\overline \va_N, r')$ in Algorithm \ref{algo:polynomial-evaluation-v1} and it runs in time $(d^m+N)\cdot\Theta(1)^m \cdot \poly(d,m,e,\log r)$ (see \autoref{thm:first-algo-ring}). 
In step 4, for each $i\in[N]$, computing the polynomial $Q_i[z]$ from $\overline b_i$ is the same as computing a base $M$ representation of the integer $\overline b_i$ of absolute value at most $r' \leq \exp(\poly(e,d,m,\log r))$. This can be done in time $\poly(d,m,e,\log r)$. From $Q_i[z]$, computing $f(\va_i)$ involves reduction modulo $r$ and $E(z)$ and this takes time $\poly(d, m, e,\log r)$. Hence, the  total cost of the step 4 is $N\cdot \poly(d,m,e,\log r)$. 

Thus, the overall time complexity of Algorithm \ref{algo:polynomial-evaluation-v1-ER} is at most 
$(d^m+N)\cdot \Theta(1)^m\cdot \poly(d,m,e,\log r)$.

\subsection{Proof of Theorem \ref{thm: main informal}}

Now, we describe the proof of \autoref{thm: main informal}.

\begin{proof}[Proof of \autoref{thm: main informal}]
Let $p$ be the characteristic of the underlying finite field $\F$ and $|\F|=p^e$. Then, we can assume that $\F$ is represented by $\modulo{p}[z]/\ideal{E(z)}$ for some irreducible monic polynomial over $\nbmodulo{p}$. Let $f(\vx)$ be the input polynomial with $m$ variables and degree less than $d$ in each variable, and $\va_1,\va_2,\ldots,\va_N\in\F^m$ be the input points. An $m$-variate polynomial of individual degree less than $d$ can have at most $(d-1)m+1$ many homogeneous components. For all $j\in\llbracket(d-1)m+1\rrbracket$, let $f_j$ be the degree $j$ homogeneous component of $f$. Each $f_j$ can be computed in time $d^m\cdot \poly(d,m)$ by counting the degree of each monomial in $f(\vx)$. Therefore, the total cost of computing all $f_j$'s is $d^m\cdot\poly(d,m)$. Now for each $j\in\llbracket(d-1)m+1\rrbracket$, we invoke Algorithm \ref{algo:polynomial-evaluation-v1-ER} with input $(f_j(\vx),\va_1,\va_2,\ldots,\va_N)$ and compute $f_j(\va_i)$ for all $i\in[N]$ in time $(d^m+N)\cdot\Theta(1)^m\cdot \poly(m,d,\log |\F|)$. Hence, the total cost of this step is bounded by $(d^m+N)\cdot\Theta(1)^m\cdot \poly(m,d,\log |\F|)$. Note that $f(\va_i)=\sum_{j=0}^{(d-1)m}f_j(\va_i)$. This implies that for each $i\in[N]$, the cost of computing $f(\va_i)$ from $f_j(\va_i)$'s is $\poly(m,d,\log|\F|)$. Therefore, the total cost incurred by this step is $N\cdot\poly(m,d,\log|\F|)$. Thus, the total time taken to compute $f(\va_i)$ for all $i\in[N]$ is $(d^m+N)\cdot\Theta(1)^m\cdot\poly(m,d,\log|\F|)$ which is $(d^m+N)^{1+o(1)}\cdot\poly(m,d,\log|\F|)$.  
\end{proof}


\section{The Second Algorithm over Rings of the Form $\nbmodulo{r}$}
\label{sec:algo-2-rings}

The main result of this section is the following theorem.

\begin{restatable}{theorem}{multimodularmain}
\label{thm:multimodular-main}
Over $\nbmodulo{r}$, for all $m\in\N$ and sufficiently large $d\in \N$, there exists a deterministic algorithm that outputs the evaluation of an $m$-variate polynomial of degree less than $d$ in each variable on $N$ points in time $$(d^m+N)^{1+o(1)}\cdot \poly(m,d,\log r),$$ provided that $\log^{\circ c}r\leq d^{o(1)}$ for some fixed constant $c\in\N$. 
\end{restatable}

We need the following lemma. It gives a way of computing the evaluation of $f(\vx)$ over a ring $R$ at a point $\va$ from the evaluations of Hasse derivatives of $f(\vx)$ at another point $\var b$, provided that the coordinates of $\va-\var b$ are in a nilpotent ideal of $R$.

\begin{lemma}\label{lem:lift-from-remainder}
Let $f(\vx)$ be an $m$-variate polynomial over a commutative ring $R$. Let $I$ be an ideal of $R$ and $s$ be a positive integer such that $I^s=0$. Let $\va=(a_1,\dots, a_m),\var b=(b_1,\dots,b_m)\in R^m$ such that $a_i\equiv b_i\pmod{I}$ for $i\in [m]$. Then
\[
f(\va)=\sum_{\ve\in\N^m: |\ve|_1<s} \hpartial_{\ve}(f)(\var b)\cdot (\va-\var b)^\ve.
\]
\end{lemma}

\begin{proof}
By the definition of Hasse derivatives,  
\[
f(\va)=f(\var b+(\va-\var b))=\sum_{\ve\in\N^m} \hpartial_{\ve}(f)(\var b)\cdot (\va-\var b)^\ve.
\]
The above sum is well-defined since $ \hpartial_{\ve}(f)=0$ when $|\ve|_1$ is sufficiently large. The lemma follows by noting that $(\va-\var b)^\ve$ is in the ideal $I^{|\ve|_1}$, which is zero when $|\ve|_1\geq s$.
\end{proof}

\subsection{A Basic Algorithm}

We first describe a basic algorithm,  \textsc{MME-Product-Set},  that evaluates a polynomial $f(\vx)\in \modulo{r^s}[\vx]$ at $N$ points in $\modulo{r^s}^m$ simultaneously. Its time complexity is not good enough, but this will be improved in later subsections.

{
\centering
\begin{minipage}{\algwidth}
\begin{algorithm}[H]
\caption{Basic Algorithm}
\label{algo:polynomial-evaluation-v2-basic}
Algorithm \textsc{MME-Product-Set}$(f(\vx),\va_1,\va_2,\ldots, \va_N, r, s)$

\medskip
\noindent where $f(\vx)$ is an $m$-variate polynomial over $\nbmodulo{r^s}$ of individual degree at most $d-1$, $\va_1,\va_2\ldots,\va_N$ are evaluation points in $\modulo{r^s}^m$,  and $s\in [m]$.
\begin{enumerate}
    \item For all $\ve\in \N^m$ with $|\ve|_1<s$, use \autoref{lem:computing-hasse-derivation} to compute $f_{\ve}(\vx):=\hpartial_{\ve}(f)(\vx)\in \modulo{r^s}[\vx]$.
    \item For all $\ve\in\N^m$ with $|\ve|_1<s$, use \autoref{lem:FFT-over-ring} to compute $f_{\ve}(\va)\in \nbmodulo{r^s}$ for $\va\in  \llbracket r\rrbracket^m$, where $\llbracket r\rrbracket$ is identified with a subset of $\nbmodulo{r^s}$ via $i\mapsto i+r^s\Z$. 
    \item For all $i\in[N]$, compute $\bar{\va}_i\in \llbracket r\rrbracket^m\subseteq \modulo{r^s}^m$ such that the coordinates of  $\bar{\va}_i$ are the remainders of the corresponding coordinates of $\va_i$ modulo $r$.
  \item   For all $i\in [N]$, compute $f(\va_{i})$ by
    \begin{equation}\label{eq:eval}
    f(\va_i)= \sum_{\ve\in\N^m: |\ve|_1<s} f_\ve(\bar{\va}_i)\cdot (\va_i-\bar{\va}_i)^\ve \in \nbmodulo{r^s}
    \end{equation}
    and output it.
    \end{enumerate}
\end{algorithm}
\end{minipage}
}

\begin{lemma}\label{lem:basic-algorithm}
Given the input $(f(\vx),\va_1,\va_2,\ldots, \va_N, r, s)$, the algorithm \textsc{MME-Product-Set} computes $f(\va_i)$ for all $i\in [N]$ in time  $O(\binom{m+s-1}{s-1}(d^m+r^m+N)) \cdot \poly(m,d,\log r)$.
\end{lemma}

\begin{proof}
Let $I=r\Z/r^s\Z$. Then $I^s=0$ and the coordinates of $\va_i-\bar{\va}_i$ are all in $I$.
So \eqref{eq:eval} holds by  \autoref{lem:lift-from-remainder}. This shows that the algorithm correctly computes $f(\va_i)$ for $i\in [N]$.

Step 1 takes time $O(\binom{m+s-1}{s-1}d^m)\cdot\poly(m,d,\log r)$ by \autoref{lem:computing-hasse-derivation}.
And Step 2 takes time  $O(\binom{m+s-1}{s-1}(d^m+r^m)) \cdot \poly(m,d,\log r)$ by \autoref{lem:FFT-over-ring}.
Step 3 takes time $O(N)\cdot\poly(m,\log r)$.
Finally, Step 4 takes time $O(\binom{m+s-1}{s-1}N) \cdot \poly(m,\log r)$.
So the total time complexity is $O(\binom{m+s-1}{s-1}(d^m+r^m+N)) \cdot \poly(m,d,\log r)$.
\end{proof}

\subsection{The Description of the Algorithm}

We describe the second algorithm \textsc{MME-B} now.

{
\centering
\begin{minipage}{\algwidth}
\begin{algorithm}[H]
\caption{The Second Algorithm over $\nbmodulo{r^s}$}
\label{algo:polynomial-evaluation-v2}
Algorithm \textsc{MME-B}$(f(\vx),\va_1,\va_2,\ldots, \va_N, r,s,t)$

\medskip
\noindent where $f(\vx)$ is an $m$-variate polynomial over $\nbmodulo{r^s}$ of individual degree at most $d-1$, $\va_1,\va_2\ldots,\va_N$ are evaluation points in $\modulo{r^s}^m$, $s\in [m]$, and $t\geq 0$ is the depth of the reduction tree.
\begin{enumerate}
    \item If $t=0$, invoke \textsc{MME-Product-Set} with input $(f(\vx),\va_1,\va_2,\ldots, \va_N, r, s)$ to compute $f(\va_i)$ for $i\in [N]$, and return.
    \item For all $\ve\in\N^m$ with $|\ve|_1<s$, use \autoref{lem:computing-hasse-derivation} to compute $f_\ve(\vx):=\hpartial_{\ve}(f)(\vx)\in\modulo{r^s}[\vx]$, and then compute a lift $\tilde{f}_\ve(\vx)\in\Z[\vx]$ of $f_\ve(\vx)$ with coefficients in $\llbracket r^s\rrbracket$. 
    \item For all $i\in[N]$, compute $\tilde{\va}_i\in \llbracket r\rrbracket^m$ such that the coordinates of  $\tilde{\va}_i$ are the remainders of the corresponding coordinates of  $\va_i$ modulo $r$, and compute $\bar{\va}_i:=\tilde{\va}_i\bmod r^s \in\modulo{r^s}^m$.
    \item Let $M:=d(r-1)^d$. Find primes $p_1<p_2<\dots<p_k$ less than or equal to $16\log M$ such that $\prod_{j=1}^k p_j>M$.
    \item For  all $\ve\in\N^m$ with $|\ve|_1<s$ and $j\in[k]$, compute $f_{\ve, j}(\vx):=\tilde{f}_\ve(\vx) \bmod p_j^m \in \modulo{p_j^m}[\vx]$.
    \item For all $i\in[N]$ and $j\in[k]$, compute $\va_{i,j}:= \tilde \va_i \bmod p_j^m \in \modulo{p_j^m}^m$.
    
    \item For all $\ve\in\N^m$ with $|\ve|_1<s$ and $j\in[k]$, invoke \textsc{MME-B} with input $(f_{\ve, j}, \va_{1,j},\va_{2,j},\ldots, \va_{N,j}, p_j, m, t-1)$ to compute $f_{\ve,j}(\va_{i,j})\in\nbmodulo{p_j^m}$ for  $i\in[N]$.
    
    
    \item For all  $\ve\in\N^m$ with $|\ve|_1<s$ and $i\in[N]$, use the Chinese Remainder Theorem (\autoref{thm:CRT}) to compute $\tilde{f}_\ve(\tilde \va_i)$ as the unique $Q_i\in \left\llbracket\prod_{j=1}^k p_j^m\right\rrbracket$ such that $Q_i\bmod p_j^m=f_{\ve,j}(\va_{i,j})$ for $j\in [k]$, and then compute $f_\ve(\bar{\va}_i)=\tilde{f}_\ve(\tilde \va_i)\bmod r^s\in  \nbmodulo{r^s}$.  
    \item  For all $i\in[N]$, compute and output 
    \begin{equation}\label{eq:eval2}
    f(\va_i)= \sum_{\ve\in\N^m: |\ve|_1<s} f_\ve(\bar{\va}_i)\cdot (\va_i-\bar{\va}_i)^\ve \in \nbmodulo{r^s}.
    \end{equation}
\end{enumerate}
\end{algorithm}
\end{minipage}
}

\subsection{The Correctness of Algorithm \ref{algo:polynomial-evaluation-v2}}

\label{subsec:correctness-algo2}

We prove the correctness of the algorithm \textsc{MME-B} (Algorithm \ref{algo:polynomial-evaluation-v2}), as stated by the following claim.

\begin{claim}
\label{claim:correctness-algo2}
Given the input $(f(\vx),\va_1,\va_2,\ldots, \va_N, r, s, t)$, the algorithm  \textsc{MME-B} computes $f(\va_i)$ for all $i\in[N]$.
\end{claim}

\begin{proof}
 We prove the claim via induction on $t$. When $t=0$, the algorithm invokes  \textsc{MME-Product-Set} in Step 1 to compute  $f(\va_i)$ for $i\in[N]$ and the claim holds by \autoref{lem:basic-algorithm}.
 
 Now consider $t\geq 1$ and assume the claim holds for $t'=t-1$. 
 
 In Step 2, we compute the Hasse derivatives $f_\ve(\vx)=\hpartial_{\ve}(f)(\vx)$ of $f(\vx)$ and then lift them to $\tilde{f}_\ve(\vx)\in\Z[\vx]$.  In Step 3, we compute $\tilde{\va}_i\in \llbracket r\rrbracket^m$ from $\va_i$ and then compute $\bar{\va}_i:=\tilde{\va}_i\bmod r^s\in\modulo{r^s}^m$. In Step 4, we compute the primes $p_j$, whose existence follows from \autoref{lem:primes}. In Step 5 and Step 6, we compute $f_{\ve, j}(\vx):=\tilde{f}_\ve(\vx) \bmod p_j^m$ and $\va_{i,j}:= \tilde \va_i \bmod p_j^m$ respectively. 
 
 In Step 7, we invoke  \textsc{MME-B} with input $(f_{\ve, j}, \va_{1,j},\va_{2,j},\ldots, \va_{N,j}, p_j, m, t')$ where $t'=t-1$. By the induction hypothesis, this correctly returns $f_{\ve,j}(\va_{i,j})\in\nbmodulo{p_j^m}$ for  $i\in[N]$.
 
 In Step 8, we use the Chinese Remainder Theorem (\autoref{thm:CRT}) to compute the unique $Q_i\in \left\llbracket\prod_{j=1}^k p_j^m\right\rrbracket$ for $i\in [N]$ such that 
 \[
 Q_i\bmod p_j^m=f_{\ve,j}(\va_{i,j}) \quad \text{for $j\in [k]$}.
 \]
We claim $\tilde{f}_\ve(\tilde \va_i)=Q_i$.
As $f_{\ve, j}(\vx)=\tilde{f}_\ve(\vx) \bmod p_j^m$ and $\va_{i,j}:= \tilde \va_i \bmod p_j^m$, we do have 
\[
\tilde{f}_\ve(\tilde \va_i)\bmod p_j^m=f_{\ve,j}(\va_{i,j}) \quad \text{for $j\in [k]$}.
\]
To prove the claim, it remains to show  that  $\tilde{f}_\ve(\tilde \va_i)\in  \left\llbracket\prod_{j=1}^k p_j^m\right\rrbracket$. As $\tilde{\va}_i\in \llbracket r\rrbracket^m$,  $\tilde{f}_\ve(\vx)$ is a polynomial of total degree at most $(d-1)m$ with at most $d^m$ monomials, and the coefficients of these monomials are in $\llbracket r^s\rrbracket$, we see that $\tilde{f}_\ve(\tilde \va_i)$ is a non-negative integer bounded by $d^m \cdot (r^s-1) \cdot (r-1)^{(d-1)m}<d^m r^{dm}<\prod_{j=1}^k p_j^m$. Here we use the facts  $s\leq m$ and $\prod_{j=1}^k p_j>M=d(r-1)^d$. So $\tilde{f}_\ve(\tilde \va_i)=Q_i$.

We also compute $f_\ve(\bar{\va}_i)=\tilde{f}_\ve(\tilde \va_i)\bmod r^s\in  \nbmodulo{r^s}$ in Step 8. This equality holds since $\tilde{f}_\ve(\vx)\bmod r^s=f_\ve(\vx)$ and $\tilde \va_i\bmod r^s=\bar{\va}_i$.

Finally, in Step 9, we compute $f(\va_i)$ from the evaluations $f_\ve(\bar{\va}_i)$ for $i\in [N]$ using \eqref{eq:eval2}.
Let $I=r\Z/r^s\Z$. Then $I^s=0$ and the coordinates of $\va_i-\bar{\va}_i$ are all in $I$ by the choice of $\bar{\va}_i$.
So \eqref{eq:eval} holds by  \autoref{lem:lift-from-remainder}. This proves the correctness of the algorithm.
\end{proof}

\subsection{The Time Complexity of Algorithm \ref{algo:polynomial-evaluation-v2}}
\label{subsec:time-complexity-algo2}
We now analyze the time complexity of Algorithm \ref{algo:polynomial-evaluation-v2}. 
When $t=0$, the algorithm only executes Step 1, and its time complexity is 
$O(\binom{m+s-1}{s-1}(d^m+r^m+N)) \cdot \poly(m,d,\log r)
$ by \autoref{lem:basic-algorithm}. 

Now assume $t\geq 1$.
Step 2 takes time $O(\binom{m+s-1}{s-1}d^m)\cdot\poly(m,d,\log r)$ by \autoref{lem:computing-hasse-derivation}.
Step~3 takes time $O(N)\cdot\poly(m,\log r)$.  In Step 4, we compute the primes $p_j$ using the Sieve of Eratosthenes \cite[\defaultS5.4]{S09}, which takes time $\tilde{O}(\log M)\leq \poly(d,\log r)$.

Using  $k\leq p_k=O(\log M)=O(d\log r)$, we see that Step 5 takes time $O(\binom{m+s-1}{s-1}d^m)\cdot \poly(m,d,\log r)$, and Step 6 takes time $O(N)\cdot \poly(m,d,\log r)$.
Step 8 takes time $O(\binom{m+s-1}{s-1}N)\cdot\poly(m,d,\log r)$ by \autoref{thm:CRT}.
Step 9 takes time $O(\binom{m+s-1}{s-1}N)\cdot\poly(m,d,\log r)$.

So the total time complexity of Steps 2--6 and 8--9 is 
\[
O\left(\binom{m+s-1}{s-1}(d^m+N)\right)\cdot\poly(m,d,\log r).
\]
Let $T(r,s,t)$ be the time complexity of the algorithm. We have 
\begin{equation}\label{eq:relation1}
T(r,s,0)=O\left(\binom{m+s-1}{s-1}(d^m+r^m+N)\right) \cdot \poly(m,d,\log r).
\end{equation}
And for $t\geq 1$,
\begin{equation}\label{eq:relation2}
T(r,s,t)\leq O\left(\binom{m+s-1}{s-1}\log M\right) T(r',m,t-1) +O\left(\binom{m+s-1}{s-1}(d^m+N)\right)\poly(m,d,\log r).
\end{equation}
where $r'=O(\log M)=O(d\log r)$.

For convenience, we define the function
\[
\lambda_k(x):=\prod_{i=0}^{\min\{k, \log^\star x-1\}} \log^{\circ i}x.
\] 
Observe that $ \lambda_k(x)=x^{1+o(1)}$ and that $\lambda_i(x)\leq \lambda_j(x)$ for $i\leq j$ and $x\geq 1$.

\begin{claim}\label{claim:technical-bound}
Let $r_0=r\geq 2$ and $r_i=c d\log r_{i-1}$ for $i\geq 1$, where $c\geq 1$ is a constant.
Then $r_k\leq c'\lambda_k(d)\log^{\circ k} r$ for all $0\leq k<\log^{\star} r$ and a sufficiently large constant $c'\geq 1$.
\end{claim}
\begin{proof}
Induct on $k$. The claim holds obviously holds for $k=0$. Now let $k\geq 1$. Assume the claim holds for $k'=k-1$.
Then 
\[
r_k=c d\log r_{k-1}\leq c d\log (c'\lambda_{k-1}(d)\log^{\circ (k-1)}r)=cd\log c'+cd\log(\lambda_{k-1}(d))+cd\log^{\circ k}r.
\]
As $c'$ is sufficiently large, we may assume $cd\log c'+(c'/2)\lambda_k(d)+cd\log^{\circ k} r\leq c'\lambda_k(d)\log^{\circ k}r$.
It remains to prove $cd\log(\lambda_{k-1}(d))\leq (c'/2)\lambda_k(d)$.
Here 
\[
\log(\lambda_{k-1}(d))=\sum_{i=1}^{\min\{k, \log^\star d\}} \log^{\circ i}d
\quad\text{and}\quad
\lambda_k(d)/d=\prod_{i=1}^{\min\{k, \log^\star d-1\}} \log^{\circ i}d
\]
Repeatedly applying the fact that $ab\geq a+b$ when $a,b\geq 2$ shows that $c\log(\lambda_{k-1}(d))\leq (c'/2)\lambda_k(d)/d$ for a sufficiently large constant $c'$, and hence $cd\log(\lambda_{k-1}(d))\leq (c'/2)\lambda_k(d)$.
\end{proof}
By \autoref{claim:technical-bound},  at the $k$-th level of the recursion tree (where the top level is regarded as the first level), we have $\log M=O(\lambda_k(d)\log^{\circ k} r)\leq O(\lambda_t(d)\log^{\circ k} r)$.
Using this to solve the  recurrence relations \eqref{eq:relation1} and \eqref{eq:relation2}, and noting that $\binom{m+s-1}{s-1}\leq 2^{O(m)}$ for $s\leq m$ (see \autoref{prop:binomial-estimation}), we obtain
\[
T(r,s,t)= O\left(\binom{m+s-1}{s-1}2^{Ctm}\lambda_t(d)^t\lambda_{t-1}(\log r) \left((\lambda_t(d)\log^{\circ t} r)^m+N\right)\right)  \cdot \poly(m,d,\log r)
\]
for $0\leq t<\log^{\star}(r)$, where $C$ is a large enough absolute constant.

Combining the above time complexity analysis with \autoref{claim:correctness-algo2} yields the following theorem.

\begin{theorem}
\label{thm:multimodular}
Let $f(\vx)$ be an $m$-variate polynomial over $\nbmodulo{r^s}$ of individual degree less than $d$, where $s\leq m$. Let $\va_1,\va_2,\ldots, \va_N$ be $N$ points in $\modulo{r}^m$. Let $t$ be a non-negative integer less than $\log^\star r$. Then, given $(f, \va_1,\va_2,\ldots,\va_N, r, s, t)$, the algorithm \textsc{MME-B} computes $f(\va_i)$ for all $i\in[N]$ in time
\[
O\left(\binom{m+s-1}{s-1}2^{Ctm}\lambda_t(d)^t\lambda_{t-1}(\log r) \left((\lambda_t(d)\log^{\circ t} r)^m+N\right)\right)  \cdot \poly(m,d,\log r).
\] 
where $C>0$ is a large enough absolute constant.
\end{theorem}

Now we are ready to prove \autoref{thm:multimodular-main}. For convenience, we restate the theorem.

\multimodularmain*

\begin{proof}
We invoke the algorithm \textsc{MME-B} for $t=\min\{c,\log^\star r-1\}$ and $s=1$.
Then, from \autoref{thm:multimodular}, we get the evaluations in time 
\[
O\left(2^{Ctm}\lambda_t(d)^t\lambda_{t-1}(\log r)((\lambda_t(d)\log^{\circ t}r)^m+N)\right)\cdot \poly(m,d,\log r).
\]
where $C>0$ is a large enough absolute constant.
From the definition, $\lambda_t(d)=d^{1+o(1)}$. Also, $\log^{\circ t} r=d^{o(1)}$. Therefore, $\lambda_t(d)\log^{\circ t}r=d^{1+o(1)}$. Since $t$ is bounded by some constant, $2^{Ctm}\lambda_t(d)^t=2^{O(m)}\cdot \poly(d)$. Also, $\lambda_{t-1}(\log r)=(\log r)^{1+o(1)}$. Thus, the overall time complexity is bounded by
\[
(d^m+N)^{1+o(1)}\cdot\poly(m,d,\log r).
\]
\end{proof}

\section{The Second Algorithm over Extension Rings}
\label{sec:algo-2-extn}

The main result of this section is the following theorem.
\begin{restatable}{theorem}{multimodularextmain}
\label{thm:multimodular-ext-main}
Over a ring $R=\modulo{r}[z]/(E(z))$, where $E(z)\in (\Z/r\Z)[z]$ is a monic polynomial of degree $e\geq 1$, for all $m\in\N$ and sufficiently large $d\in \N$, there exists a deterministic algorithm that outputs the evaluation of an $m$-variate polynomial of degree less than $d$ in each variable on $N$ points in time $$(d^m+N)^{1+o(1)}\cdot \poly(m,d,\log |R|),$$ provided that $\log^{\circ c}|R|\leq d^{o(1)}$ for some fixed constant $c\in\N$.
\end{restatable}

\begin{remark*}
By choosing $r$ to be a prime number $p$ and $E(z)$ to be a monic irreducible polynomial over $\F_p$, we see that
\autoref{thm: main informal} follows from \autoref{thm:multimodular-ext-main} assuming that the size of the finite field is at most $(\exp(\exp(\exp(\cdots (\exp(d)))))$, where the height of the tower of exponentials is some constant.
\end{remark*}

We present an algorithm  \textsc{MME-For-Extension-Rings-B}, built on  the algorithm \textsc{MME-B} in the previous section, that allows us to prove \autoref{thm:multimodular-ext-main}. 
Let $R:=\modulo{r}[z]/(E(z))$.
Given an $m$-variate polynomial $f(\vx)$ over $R$ and $N$ evaluation points $\va_1,\va_2\ldots,\va_N\in R^m$, the algorithm outputs the evaluations of $f(\vx)$ at  $\va_1,\va_2\ldots,\va_N$.

Following the algorithm in \cite[\defaultS4.3]{Kedlaya11}, also presented in \autoref{sec: algo 1 extension fields}, our first step is lifting $f(\vx)$ to a polynomial $F(\vx)\in (\Z[z])(\vx)$, and similarly lifting the evaluation points $\va_1, \va_2, \dots,\va_N$ to $\tilde{\va}_1,\tilde{\va}_2,\dots,\tilde{\va}_N\in \Z[z]^m$, such that the coefficients or coordinates of the lifts are reasonably bounded. More specifically, the coefficients of $F(\vx)$ and the coordiantes of $\va_i$'s are polynomials in $\Z[z]$ with degree at most $e-1$ and the coefficients are in $\llbracket r\rrbracket$. Let $M:=d^m(e(r-1))^{(d-1)m+1}+1$. Then, Kedlaya and Umans observed that for all $i\in[N]$, $F(\tilde\va_i)$ is a polynomial in $z$ with degree at most $(e-1)dm$ and the coefficients are less than $M$. Therefore, if we know $F(\tilde\va_i)$ at $z=M$, we can compute $F(\tilde\va_i)$. Since the evaluation of $F(\tilde\va_i)$ at $z=M$ is less than $r'=M^{(e-1)dm+1}$, \cite{Kedlaya11} reduces the computing of $F(\tilde\va_i)$ to the computing of $F(\vx)\bmod(r', z-M)$. In this way, \cite{Kedlaya11} showed that we can reduce the multivariate multipoint evaluation problem over $R$  to that over $\Z[z]/(r', z-M)=\nbmodulo{r'}$, and the latter can be solved by the algorithm \textsc{MME-B} (Algorithm \autoref{algo:polynomial-evaluation-v2}).

The problem with applying  this reduction to prove  \autoref{thm:multimodular-ext-main} is that $r'$ is too large for us, being exponential in $m^2$. If we simply use this reduction, the resulting algorithm would have the claimed time complexity only when  $m$ is at most $(\exp(\exp(\exp(\cdots (\exp(d)))))$, where the height of the tower of exponentials is some constant.

To resolve this issue, we give a more efficient reduction, which leads to the algorithm  \textsc{MME-For-Extension-Rings-B} in this section. The following is a brief overview of the algorithm. 

\paragraph*{Overview.} Our goal is to compute  $F(\tilde{\va}_i)$ for each $i\in [N]$, whose remainder modulo $(r, E(z))$ would give the desired evaluation $f(\va_i)$. As mentioned above, for each $i\in[N]$, $F(\tilde\va_i)$ is a polynomial in $\Z[z]$ with degree less than $(e-1)md+1$ and the coefficients are less than $M$. In \cite{Kedlaya11}, they interpolated $F(\tilde\va_i)$ from its evaluation at a single point $z=M$. Contrary to them, we  interpolate $F(\tilde\va_i)$ from its evaluation at multiple points. Since the degree of $F(\va_i)$ is less than $(e-1)dm+1$, one way can be to compute the evaluations of $F(\va_i)$ at all the points in $\llbracket (e-1)dm+1\rrbracket$ , and then interpolate $F(\va_i)$. This would reduce the computation of $F(\va_i)$ for all $i\in[N]$ to $(e-1)dm+1$ instances of multivariate multipoint evaluation over $\Z$. Next, to solve those instances using the algorithm \textsc{MME-B} (Algorithm \autoref{algo:polynomial-evaluation-v2}), we need to reduce them as instances of multivariate multipoint evaluation over rings of form $\nbmodulo{r'^s}$ where $s\in[m]$. For proving \autoref{thm:multimodular-ext-main}, we want the value of $r'$ is independent of $m$. Since the coefficients of $F(\va_i)$ are less than $M$, it would be sufficient for us to compute $F(\va_i)\bmod M$. Since $M^{1/m}=O(d(er)^d)$, which is independent of $m$,  one can pick $r'$ as an integer greater than $M^{1/m}$ and try to work over the ring $\nbmodulo{r'^m}$. However, there is a problem with this approach. To interpolate $F(\tilde\va_i)\bmod r'^m$ from its evaluations at the points in $\llbracket(e-1)dm+1\rrbracket$, we need to ensure that $(j-j')$ is a unit in the ring $\nbmodulo{r'^m}$ for distinct $j,j'\in\llbracket (e-1)dm+1\rrbracket$. One possible way to satisfy both the constraints is by picking $r'$ as a prime power greater than $M^{1/m}$ for some prime larger than $(e-1)dm$. This imposes that $r'$ has to be at least $(e-1)dm$, which makes it dependent on $m$. To overcome this issue, we do the following.
\begin{enumerate}
    \item Let $\ell=ed$ and $P$ be a prime in $[\ell, 2\ell]$. Then, we pick $r'$ as a $P$-power greater than $M^{1/m}$.
    \item For every $j\in\llbracket\ell\rrbracket$, we compute $F(\tilde\va_i)\bmod(r'^m,(z-j)^m)$.
\end{enumerate}
The first condition ensures that $r'$ is independent of $m$ and for distinct $j,j'\in\llbracket\ell\rrbracket$, $(j-j')$ is a unit in $\nbmodulo{r'^m}$. Therefore, from $F(\va_i)\bmod(r'^m,(z-j)^m)$ for all $j\in\llbracket\ell\rrbracket$, we get $F(\tilde\va_i)$ modulo $r'^m$ using Hermite interpolation (\autoref{lem:hermitian-interpolation}) over $\nbmodulo{r'^m}$. This is sufficient to compute $F(\tilde\va_i)$, since the coefficients of $F(\tilde\va_i)$ are less than $r'^m$.

How do we compute $F(\tilde{\va}_i) \bmod (r'^m, (z-j)^m)$? 
Consider the Taylor expansion $\tilde{\va}_i\equiv\sum_{u=0}^{e-1} \va_{i,j,u} (z-j)^u \bmod r'^m$ where the coefficients $\va_{i,u}\in\nbmodulo{r'^m}$. Let $F(\vx)\equiv\sum_{u=0}^{e-1} f_u(\vx) z^u \bmod r'^m$ where $f_u(\vx)\in\modulo{r'^m}[\vx]$ is the coefficient of $z^u$ in $F(\vx)\bmod r'^m$.  This implies that $F(\tilde{\va}_i) \bmod (r'^m, (z-j)^m)$ is same as the $\sum_{i=1}^{e-1} (f_u(\tilde\va_i)z^u \bmod (z-j)^m)$ since the choice of $r'$ ensures that $\tilde\va_i\bmod r'^m$ is same as $\tilde\va_i$. From \autoref{lem:lift-from-remainder}, $f_u(\tilde \va_i)z^u \bmod (z-j)^m$ can be computed from the evaluations of $f_{\ve,u}(\vx)$ at $\va_{i,j,0}$ for all $\ve\in\N^n$ with $|\ve|_1<m$. This shows that the evaluation of $F(\tilde\va_i)\bmod (r'^m,(z-j)^m)$  can be computed from the evaluations of $f_{\ve, u}(\vx)$ at $\va_{i,j,0}$ for all $\ve\in\N^n$ with $|\ve|_1<m$. Thus, for every $\ve\in\N^m$ with $|\ve|_1<m$, $j\in\llbracket\ell\rrbracket$ and $u\in\llbracket e\rrbracket$, we need to solve the following instance of multivariate multipoint evaluation $(f_{\ve, u}, \va_{1,j,0},\va_{2,j,0},\ldots,\va_{N,j,0})$ over the ring $\nbmodulo{r'^m}$. The latter problem can be solved using the algorithm \textsc{MME-B} (Algorithm \autoref{algo:polynomial-evaluation-v2}).

\subsection{The Description of the Algorithm}

Now, we formally describe the algorithm \textsc{MME-For-Extension-Rings-B}.

{
\centering
\begin{minipage}{\algwidth}
\begin{algorithm}[H]
\caption{The Second Algorithm over Extension Rings}
\label{algo:polynomial-evaluation-v2-ext}
Algorithm \textsc{MME-For-Extension-Rings-B}$(f(\vx),\va_1,\va_2,\ldots, \va_N, r,t,E(z))$

\medskip
\noindent where $f(\vx)$ is an $m$-variate polynomial over $R:=\modulo{r}[z]/(E(z))$ of individual degree at most $d-1$, $E(z)\in (\Z/r\Z)[z]$ is a monic polynomial of degree $e\geq 1$, $\va_1,\va_2\ldots,\va_N$ are evaluation points in $R^m$, and $t\geq 0$ is the depth of the reduction tree.
\begin{enumerate}
    \item Compute $F(\vx)\in (\Z[z])[\vx]$ as a lift of $f(\vx)$ such that every coefficient of $F(\vx)$ is a polynomial in $z$ of degree at most $e-1$ with coefficients in $\llbracket r\rrbracket$.
    \item For all $i\in[N]$, compute $\tilde{\va}_i\in \Z[z]^m$ as a lift of $\va_i$ such that every coordinate of $\tilde{\va}_i$ is  a polynomial in $z$ of degree at most $e-1$ with coefficients in $\llbracket r\rrbracket$.
    \item Let $\ell=ed$.  Choose a prime $P\in [\ell, 2\ell]$. Choose the smallest $P$-power $r'$ such that $r'\geq d(er)^{d}$.  Compute $\bar{F}(\vx):=F(\vx)\bmod r'^m\in \nbmodulo{r'^m}$. For $i\in[N]$, compute $\bar{\va}_i:=\tilde{\va}_i  \bmod r'^m \in \modulo{r'^m}[z]^m$.
    \item 
    Compute $f_{0}(\vx),f_{1}(\vx)\dots, f_{e-1}(\vx)\in \modulo{r'^m}[\vx]$  such that
$\bar{F}(\vx)=\sum_{u=0}^{e-1}  f_{u}(\vx) z^u$.
    For all $i\in[N]$ and $j\in\llbracket \ell\rrbracket$, compute $\va_{i,j,0},\va_{i,j,1},\dots,\va_{i,j,e-1}\in \modulo{r'^m}^m$ such that $\bar{\va}_i =\sum_{u=0}^{e-1} \va_{i,j,u} (z-j)^u$.
    \item For all $\ve\in \N^m$ with $|\ve|_1<m$ and $u\in\llbracket e\rrbracket$, use \autoref{lem:computing-hasse-derivation} to compute $f_{\ve,u}(\vx):=\hpartial_{\ve}(f_{u})(\vx)\in \modulo{r'^m}[\vx]$.
    \item For all $\ve\in\N^m$ with $|\ve|_1<m$,  $j\in\llbracket \ell\rrbracket$, and $u\in\llbracket e\rrbracket$, invoke \textsc{MME-B} with input $(f_{\ve, u}, \va_{1,j,0},\va_{2,j,0},\ldots, \va_{N,j,0}, r', m, t)$ to compute $f_{\ve,u}(\va_{i,j,0})\in \nbmodulo{r'^m}$ for all $i\in[N]$. 
%
    \item For all $i\in [N]$ and $j\in \llbracket \ell\rrbracket$, compute 
    \[
    b_{i,j}(z)=\left(\sum_{u=0}^{e-1} \sum_{\mathbf{e}\in\N^m: |\mathbf{e}|_1<m}  f_{\ve,u}(\va_{i,j,0})  (\bar{\va}_i-\va_{i,j,0})^\ve z^u   \right) \bmod (z-j)^m 
    \]
    which is a polynomial of degree at most $m-1$ in $z$ over  $\nbmodulo{r'^m}$.
    \item For all $i\in [N]$, use Hermite interpolation (\autoref{lem:hermitian-interpolation}) to compute the unique polynomial $b_i(z)\in \modulo{r'^m}[z]$ of degree less than $\ell m$ such that $b_i(z) \equiv b_{i,j}(z) \pmod {(z-j)^m}$ for $j\in\llbracket \ell\rrbracket$. 

    \item For all $i\in[N]$, lift $b_i(z)$ to $B_i(z)\in \Z[z]$ with coefficients in $\llbracket r'^m\rrbracket$, and output the remainder of $B_i(z)$ modulo $(r, E(z))$ as $f(\va_i)$. 
    \end{enumerate}
\end{algorithm}
\end{minipage}
}

\subsection{The Correctness of Algorithm \ref{algo:polynomial-evaluation-v2-ext}}

We prove the correctness of the algorithm \textsc{MME-For-Extension-Rings-B} (Algorithm \ref{algo:polynomial-evaluation-v2-ext}), as stated by the following claim.

\begin{claim}\label{claim:correctness-algo2-ext}
Given the input $(f(\vx),\va_1,\va_2,\ldots, \va_N, r, t, E(z))$, the algorithm \textsc{MME-For-Extension-Rings-B} computes $f(\va_i)$ for all $i\in[N]$.
\end{claim}
 
\begin{proof}

The first two steps of the algorithm compute the lifts $F(\vx)$ and $\tilde{\va}_i$ for $i\in [N]$.
As shown in \autoref{sec:correctness-poly-eval-1-ER},  for each $i\in [N]$, the degree of $F(\tilde{\va}_i)\in\Z[z]$ is bounded by $(e-1)md$, and the coefficients of $F(\tilde{\va}_i)$ are non-negative integers less than $M:=d^m(e(r-1))^{(d-1)m+1}+1$.

Next, we let $\ell=ed$ and compute a prime $P\in [\ell, 2\ell]$, which can be done by Bertrand's postulate.
Then we find the smallest $P$-power $r'$ such that $r'\geq d(er)^d$, which guarantees $r'^m\geq d^m(er)^{dm}\geq M$.
Next, we compute  $\bar{F}(\vx):=F(\vx)\bmod r'^m$ and for $i\in [N]$, compute $\bar{\va}_i:=\tilde{\va}_i  \bmod r'^m \in \modulo{r'^m}[z]^m$, so that $\bar{F}(\bar{\va}_i)=F(\tilde{\va}_i) \bmod r'^m$. As  the coefficients of $F(\tilde{\va}_i)$ are non-negative integers less than $M\leq r'^m$, to compute $F(\tilde{\va}_i)$, we just need to first compute $\bar{F}(\bar{\va}_i)$ and then lift its coefficients to integers in $\llbracket r'^m\rrbracket$. This is precisely what the remaining steps (Steps 4--9) do.

In Step 4, we compute the data $f_{u}(\vx)\in  \modulo{r'^m}[\vx]$ and  $\va_{i,j,u}\in  \modulo{r'^m}^m$ such that $\bar{F}(\vx)=\sum_{u=0}^{e-1}  f_{u}(\vx) z^u$ and $\bar{\va}_i =\sum_{u=0}^{e-1} \va_{i,j,u} (z-j)^u$. This is possible as the coefficients of $\bar{F}(\vx)$ and the coordinates of each $\bar{\va}_i$ are polynomials of degree at most $e-1$ in $z$ over  $\nbmodulo{r'^m}$. Then in Step 5, we compute the data  $f_{\ve,u}(\vx):=\hpartial_{\ve}(f_{u})(\vx)\in \modulo{r'^m}[\vx]$ for all $\ve\in\N^m$ with $|\ve|_1<m$. And in Step 6, we compute the evaluations $f_{\ve,u}(\va_{i,j,0})\in \nbmodulo{r'^m}$ using the algorithm  \textsc{MME-B}.

Consider $i\in [N]$ and $j\in \llbracket \ell\rrbracket$. 
For $u\in\llbracket e\rrbracket$, as $\bar{\va}_i-\va_{i,j,0}=\sum_{u=1}^{e-1}\va_{i,j,u} (z-j)^u$ is a multiple of $z-j$, \autoref{lem:lift-from-remainder} gives
\[
f_{u}(\bar{\va}_i)\equiv\sum_{\mathbf{e}\in\N^m: |\mathbf{e}|_1<m}  f_{\ve,u}(\va_{i,j,0})  (\bar{\va}_i-\va_{i,j,0})^\ve  \pmod{(z-j)^m}.
\]
Therefore,
\[
\bar{F}(\bar{\va}_i)=\sum_{u=0}^{e-1}  f_{u}(\bar{\va}_i) z^u \equiv\sum_{u=0}^{e-1} \sum_{\mathbf{e}\in\N^m: |\mathbf{e}|_1<m}  f_{\ve,u}(\va_{i,j,0})  (\bar{\va}_i-\va_{i,j,0})^\ve z^u  \pmod{(z-j)^m}.
\]
We compute $b_{i,j}(z):=\bar{F}(\bar{\va}_i)\bmod (z-j)^m$ in Step 7 using the above equation.

In Step 8, we compute $b_i(z)=\bar{F}(\bar{\va}_i)$ from its remainders $b_{i,j}$ modulo $(z-j)^m$ for $i\in [N]$ and $j\in\llbracket\ell\rrbracket$ using Hermite interpolation (\autoref{lem:hermitian-interpolation}). As  $\ell=ed$, we have $\deg_z(\bar{F}(\bar{\va}_i))\leq (e-1)md<\ell m$. And as $r'$ is a power of a prime $P$ and $P\geq \ell$, the difference $j-j'$ has a multiplicative inverse in $\nbmodulo{r'^m}$ for distinct $j,j'\in\llbracket \ell\rrbracket$. So $b_i(z)$ can indeed be found using Hermite interpolation. 

Finally, in Step 9, we compute the lift $F(\tilde{\va}_i)$ from $\bar{F}(\bar{\va}_i)$, and then output $f(\va_i)=F(\tilde{\va}_i)\bmod (r, E(z))$, as desired.
\end{proof}

\subsection{The Time Complexity of Algorithm \ref{algo:polynomial-evaluation-v2-ext}}

We now analyze the time complexity of Algorithm \ref{algo:polynomial-evaluation-v2-ext}. 
Step 1 takes time $O(d^m)\cdot\poly(m,d,\log|R|)$.
And Step 2 takes time $O(N)\cdot\poly(m,\log|R|)$.
Note $\ell, \log r'\leq\poly(d, \log |R|)$.
Then Step~3 and Step~4 both take time $O(d^m+N)\cdot\poly(m,d,\log|R|)$.
By \autoref{lem:computing-hasse-derivation}, Step 5 takes time $O(\binom{2m-1}{m-1}d^m)\cdot\poly(m,d,\log |R|)$.
Step 7 takes time $O(\binom{2m-1}{m-1}N)\cdot\poly(m,d,\log |R|)$.
By \autoref{lem:hermitian-interpolation}, Step 8 takes time $O(N)\cdot\poly(m,d,\log |R|)$.
And Step 9 takes time $O(N)\cdot\poly(m,\log |R|)$.

Finally, by \autoref{thm:multimodular}, the time complexity of Step 6 is bounded by 
\[
O\left(2^{C(t+1)m}\lambda_t(d)^t\lambda_{t-1}(\log r') \left((\lambda_t(d)\log^{\circ t} r')^m+N\right)\right)  \cdot \poly(m,d,\log |R|).
\] 
for $0\leq t<\log^{\star}(r')$, where $r'\in [d(er)^d, 2ed^2(er)^d]$ is as in the algorithm and $C>0$ is a large enough absolute constant. This step dominates the time complexity of the whole algorithm.

Combining the above time complexity analysis with \autoref{claim:correctness-algo2-ext} yields the following theorem.

\begin{theorem}
\label{thm:multimodular-ext}
Let $f(\vx)$ be an $m$-variate polynomial over $R=\modulo{r}[z]/(E(z))$ of individual degree less than $d$, where $E(z)$ is a monic polynomial of degree $e\geq 1$. Let $\va_1,\va_2,\ldots, \va_N$ be $N$ points in $R^m$. Let $t$ be a non-negative integer less than $\log^\star(d(er)^d)$. Then, given $(f, \va_1,\va_2,\ldots,\va_N, r, t, E(z))$, the algorithm \textsc{MME-For-Extension-Rings-B} computes $f(\va_i)$ for all $i\in[N]$ in time
\[
O\left(2^{C(t+1)m}\lambda_t(d)^t\lambda_{t-1}(\log r') \left((\lambda_t(d)\log^{\circ t} r')^m+N\right)\right)  \cdot \poly(m,d,\log |R|).
\] 
where $r'=2ed^2(er)^d$ and $C>0$ is a large enough absolute constant.
\end{theorem}

Now we are ready to prove \autoref{thm:multimodular-ext-main}. For convenience, we restate it here.

\multimodularextmain*
\begin{proof}
We invoke the algorithm  \textsc{MME-For-Extension-Rings-B} for $t=\max\{c, 2\}$, which is less than $\log^\star(d(er)^d)$ as $d$ is sufficiently large.
Then, from \autoref{thm:multimodular-ext}, we get the evaluations in time 
\[
O\left(2^{C(t+1)m}\lambda_t(d)^t\lambda_{t-1}(\log r') \left((\lambda_t(d)\log^{\circ t} r')^m+N\right)\right)  \cdot \poly(m,d,\log |R|).
\] 
where $r'=2ed^2(er)^d$ and $C>0$ is a large enough absolute constant.
From the definition, $\lambda_t(d)=d^{1+o(1)}$. 
Noting that $\log\log r'=O(\log d+\log\log |R|)$, $\log^{\circ c}|R|\leq d^{o(1)}$, and $t=\max\{c,2\}$, we have $\log^{\circ t} r'=d^{o(1)}$. 
Therefore, $\lambda_t(d)\log^{\circ t}r'=d^{1+o(1)}$. Since $t$ is bounded by some constant, $2^{C(t+1)m}\lambda_t(d)^t=2^{O(m)}\cdot \poly(d)$. Also, $\lambda_{t-1}(\log r')\leq \poly(d,\log |R|)$. Thus, the overall time complexity is bounded by $$(d^m+N)^{1+o(1)}\cdot\poly(m,d,\log |R|).$$
\end{proof}



\section*{Acknowledgement}
Mrinal is thankful to Swastik Kopparty for introducing him to the question of multipoint evaluation and the work of Kedlaya--Umans \cite{Kedlaya11} and to Prahladh Harsha and Ramprasad Saptharishi  for many helpful discussions. 





\end{document}